\documentclass[runningheads]{llncs}
\usepackage{graphicx}
\usepackage{tikz}
\usetikzlibrary{patterns}
\usepackage{wrapfig}
\usepackage{amssymb}
\usepackage{amsmath}

\newcommand\map[3]{#1 : #2 \longrightarrow #3}
\newcommand\natr[3]{#1 : #2 \Longrightarrow #3}
\newcommand\pmap[2]{#1 : #2 \rightharpoonup \textbf{pSet}}
\newcommand\nat{\mathbb{N}}
\newcommand\hda{\textbf{HDA}_\textbf{L}}

\newcommand\phda{\textbf{pHDA}_\textbf{L}}
\newcommand\precub{\textbf{pCub}}
\newcommand\spine{\textbf{Spine}_\textbf{L}}
\newcommand\ps{\textbf{PS}_\textbf{L}}

\newcommand\tr{\textbf{Tr}_\textbf{L}}
\newcommand\set{\textbf{Set}}
\newcommand\unf{\text{unf}}
\newcommand\col[1]{\text{col} \, #1}
\newcommand\eqcl[1]{\langle #1 \rangle}
\newcommand\eqclp[1]{[ #1 ]}
\def\P{\mathcal{P}}
\def\M{\mathcal{M}}
\def\C{\mathcal{C}}

\definecolor{blackorkgreen}{rgb}{0.13,0.55,0.13}

\begin{document}
\title{Trees in Partial Higher Dimensional Automata\thanks{The author was supported by ERATO HASUO Metamathematics for Systems Design 27 Project (No. JPMJER1603), JST.}}
%
%
\author{J\'er\'emy Dubut\inst{1,2}}
\authorrunning{J. Dubut}
%
\institute{National Institute of Informatics, Tokyo, Japan\and
Japanese-French Laboratory for Informatics, Tokyo, Japan\\
\email{dubut@nii.ac.jp} }
\maketitle              
\begin{abstract}
In this paper, we give a new definition of partial Higher Dimension Automata using lax functors. This definition is simpler and more natural from a categorical point of view, but also matches more clearly the intuition that pHDA are Higher Dimensional Automata with some missing faces. We then focus on trees. Originally, for example in transition systems, trees are defined as those systems that have a unique path property. To understand what kind of unique property is needed in pHDA, we start by looking at trees as colimits of paths. This definition tells us that trees are exactly the pHDA with the unique path property modulo a notion of homotopy, and without any shortcuts. This property allows us to prove two interesting characterisations of trees: trees are exactly those pHDA that are an unfolding of another pHDA; and trees are exactly the cofibrant objects, much as in the language of Quillen's model structure. In particular, this last characterisation gives the premisses of a new understanding of concurrency theory using homotopy theory.

\end{abstract}

\section{Introduction}

Higher Dimensional Automata (HDA, for short), introduced by Pratt in \cite{pratt91}, are a geometric model of true concurrency. Geometric, because they are defined very similarly to simplicial sets, and can be interpreted as glueings of geometric objects, here hypercubes of any dimension. Similarly to other models of concurrency much as event structures \cite{nielsen81}, asynchronous systems \cite{shields85,bednarczyk87}, or transition systems with independence \cite{nielsen94}, they model true concurrency, in the sense that they distinguish interleaving behaviours from simultaneous behaviours. In \cite{vanglabbeek05}, van Glabbeek proved that they form the most powerful models of a hierarchy of concurrent models. In \cite{fahrenberg13}, Fahrenberg described a notion of bisimilarity of HDA using the general framework of open maps from \cite{joyal96}. If this work is very natural, it is confronted with a design problem: paths (or executions) cannot be nicely encoded as HDA. Indeed, in a HDA, it is impossible to model the fact that two actions \emph{must} be executed at the same time, or that two actions are executed at the same time but one \emph{must} start before the other. From a geometric point of view, those impossibilities are expressed by the fact that we deal with closed cubes, that is, cubes that must contain all of their faces. Motivated by those examples, Fahrenberg, in \cite{fahrenberg15}, extended HDA to partial HDA, intuitively, HDA with cubes with some missing faces. If the intuition is clear, the formalisation is still complicated to achieve: the definition from \cite{fahrenberg15} misses the point that faces can be not uniquely defined. This comes from the fact that Fahrenberg wanted to stick to the `local' definition of precubical sets, that is, that cubes must satisfy some local conditions about faces. As we will show, those local equations are not enough in the partial case. Another missed point is the notion of morphisms of partial HDA: as defined in \cite{fahrenberg15}, the natural property that morphisms map executions to executions is not satisfied. In Section \ref{sec:phda}, we address those issues by giving a new definition of partial HDA in terms of lax functors. This definition, similar to the presheaf theoretic definition of HDA, avoid the issues discussed above by considering global inclusions, instead of local equations. This illustrates more clearly the intuition of partial HDA being HDA with missing faces: we coherently replace sets and total functions by sets and partial functions. From this similarity with the original definition of HDA, we can prove that it is possible to complete a partial HDA to turn it into a HDA, by adding the missing faces, and from this completion, it is possible to define a geometric realisation of pHDA (which was impossible with Fahrenberg's definition).\\
The geometry of Higher Dimensional Automata, and more generally, of true concurrency, has been studied since Goubault's PhD thesis \cite{goubault95}. Since then, numerous pieces of work relating algebraic topology and true concurrency have been achieved (for example, see the textbooks \cite{grandis09,fajstrup16}). In particular, some attempts of defining nice homotopy theories for true concurrency (or directed topology), through the language of model structures of Quillen \cite{quillen67}, have been made by Gaucher \cite{gaucher11}, and the author \cite{dubut17}. In the second part of this paper (Sections \ref{sec:path}, \ref{sec:unfo} and \ref{sec:cofi}), we consider another point of view of this relationship between HDA and model structures. The goal is not to understand the true concurrency of HDA, that is, understanding the homotopy theory of HDA as an abstract homotopy theory, but to understand the concurrency theory of HDA. By this we mean to understand how paths (or executions) and extensions of paths can be understood using (co)fibrations (in Quillen's sense). Also, the goal is not to construct a model structure, as Quillen's axioms would fail, but to give intuitions and some preliminary formal statements toward the understanding of concurrency using homotopy theory. Using this point of view, many constructions in concurrency can be understood using the language of model structures:\\
\indent -- Open maps from \cite{joyal96} can be understood as trivial fibrations, namely weak equivalences (here, bisimulations) that have the right lifting properties with respect to some morphisms.\\
\indent -- Those morphisms are precisely extensions of executions, which means that they can be seen as cofibration generators (in the language of cofibrantly generated model structures \cite{hirschhorn03}).\\
\indent -- Cofibrations are then morphisms that have the left lifting property with respect to open maps. In particular, this allows us to define cofibrant objects as those objects whose unique morphisms from the initial object is a cofibration. In a way, cofibrant objects are those objects that are constructed by just using extensions of paths, and should correspond to trees.\\
\indent -- The cofibrant replacement is then given by canonically constructing a cofibrant object, which is weakly equivalent (here, bisimilar) to a given object. That should correspond to the unfolding.\\
The main ingredient is to understand what trees are in this context. In the case of transition systems for semantics of CCS \cite{milner80}, synchronisation trees are those systems with exactly one path from the initial state to any state. Those trees are then much simpler to reason on, but they are still powerful enough to capture any bisimulation type: by unfolding, it is possible to canonically construct a tree from a system. The goal of Sections \ref{sec:path} and \ref{sec:unfo} will be to understand how to generalise this to pHDA. In this context, it is not clear what kind of unique path property should be considered as, in general, in truly concurrent systems, we have to deal with homotopies, namely, equivalences of paths modulo permutation of independent actions. Following \cite{dubut16}, we will first consider trees as colimits of paths. This will guide us to determine what kind of unique path property is needed: a tree is a pHDA with exactly one class of paths modulo a notion of homotopy, from the initial state to any state, and without any shortcuts. This will be proved by defining a suitable notion of unfolding of pHDA. Finally, in Section \ref{sec:cofi}, we prove that those trees coincide exactly with the cofibrant objects, illustrating the first steps of this new understanding of concurrency, using homotopy theory.

\section{Fixing the definition of pHDA}
\label{sec:phda}

In this Section, we review the definitions of HDA (Section \ref{sub:hda}), the first one using face maps, and the second one using presheaves. In Section \ref{sub:fahr}, we describe the definition of partial HDA from \cite{fahrenberg15} and explain why it does not give us what we are expecting. We tackle those issues by introducing a new definition in Section \ref{sub:phda}, extending the presheaf theoretic definition, using lax functors instead of strict functors. Finally, in Section \ref{sub:refl}, we prove that HDA form a reflective subcategory of partial HDA, by constructing a completion of a partial HDA.

	\subsection{Higher Dimensional Automata}
	\label{sub:hda}

Higher Dimensional Automata are an extension of transition systems: they are labeled graphs, except that, in addition to vertices and edges, the graph structure also has higher dimensional data, expressing the fact that several actions can be made at the same time. Those additional data are intuitively cubes filling up interleaving: if $a$ and $b$ can be made at the same time, instead of having an empty square as on the left figure, with $a.b$ and $b.a$ as only behaviours, we have a full square as on the right figure, with any possible behaviours in-between. This requires to extend the notion of graph to add those higher dimensional cubical data: that is the notion of \textit{precubical sets}.

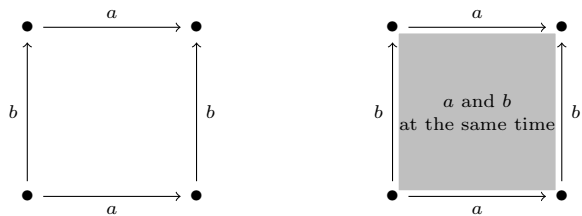
\begin{center}
\begin{tikzpicture}[scale=0.75]
		
	\node (11) at (0,0) {$\bullet$};
	\node (12) at (0,3) {$\bullet$};
	\node (21) at (3,0) {$\bullet$};
	\node (22) at (3,3) {$\bullet$};
	
%
	
	\path[->,font=\scriptsize]
		(11) edge node[left]{$b$} (12)
		(21) edge node[right]{$b$} (22)
		(11) edge node[below]{$a$} (21)
		(12) edge node[above]{$a$} (22);
			
\end{tikzpicture}
\quad\quad\quad\quad\quad
\begin{tikzpicture}[scale=0.75]
		
	\node (11) at (0,0) {$\bullet$};
	\node (12) at (0,3) {$\bullet$};
	\node (21) at (3,0) {$\bullet$};
	\node (22) at (3,3) {$\bullet$};
	
	\draw[draw = white, fill = gray!50] (0.1,0.1) rectangle (2.9,2.9);
	
	\node (c) at (1.5,1.7) {\scriptsize{$a$ and $b$}};
	\node (c') at (1.5,1.3) {\scriptsize{at the same time}};
	
	\path[->,font=\scriptsize]
		(11) edge node[left]{$b$} (12)
		(21) edge node[right]{$b$} (22)
		(11) edge node[below]{$a$} (21)
		(12) edge node[above]{$a$} (22);
			
\end{tikzpicture}
\end{center}
\vspace{-2mm}

\noindent\textbf{Concrete definition of precubical sets.} A \textbf{precubical set $X$} is a collection of sets $(X_n)_{n\in\nat}$ together with a collection of functions $(\map{\partial_{i,n}^\alpha}{X_n}{X_{n-1}})_{n>0, 1 \leq i \leq n, \alpha \in \{0,1\}}$ satisfying the local equations $\partial_{i,n}^\alpha\circ\partial_{j,n+1}^\beta = \partial_{j,n}^\beta\circ\partial_{i+1,n+1}^\alpha$ for every $\alpha, \beta \in \{0,1\}$, $n > 0$ and $1 \leq j \leq i \leq n$. A \textbf{morphism of precubical sets} from $X$ to $Y$ is a collection of functions $(\map{f_n}{X_n}{Y_n})_{n\in\nat}$ 
satisfying the equations $f_n\circ\partial_{i,n}^\alpha = \partial_{i,n}^\alpha\circ f_{n+1}$ for every $n\in\nat$, $1 \leq i \leq n$ and $\alpha \in \{0,1\}$. The elements of $X_0$ are called \textbf{points}, $X_1$ \textbf{segments}, $X_2$ \textbf{squares}, $X_n$ \textbf{$n$-cubes}. In the following, we will call \textbf{past} (resp. \textbf{future}) \textbf{$i$-face maps} the $\partial_{i,n}^0$ (resp. $\partial_{i,n}^1$). We denote this category of precubical sets by $\precub$.\\

\noindent\textbf{Precubical sets as presheaves.} Equivalently, $\precub$ is the category of presheaves over the cubical category $\Box$. $\Box$ is the subcategory of \textbf{Set} whose objects are the sets $\{0,1\}^n$ for $n \in \nat$ and whose morphisms are generated by the so-called \textbf{coface maps}:
\begin{center}
$d_{i,n}^{\alpha}: \{0,1\}^{n-1}\longrightarrow \{0,1\}^n ~~~ (\beta_1, \ldots, \beta_{n-1}) \longmapsto (\beta_1, \ldots, \beta_{i-1}, \alpha, \beta_{i}, \ldots, \beta_{n-1})$
\end{center}
A precubical set is a functor $\map{X}{\Box^{op}}{\set}$, that is, a presheaf over $\Box$, and a morphism of precubical sets is a natural transformation.\\

\noindent\textbf{Higher dimensional Automata \cite{vanglabbeek91}.} From now on, fix a set $L$, called the \textbf{alphabet}. We can form a precubical set also noted $L$ such that $L_n = L^n$ and the $i$-face maps are given by $\delta_i^\alpha(a_1\ldots a_n) = a_1\ldots a_{i-1}.a_{i+1}\ldots a_n$. We can also form the following precubical set $\ast$ such that $\ast_0 = \{\ast\}$ and $\ast_n = \varnothing$ for $n > 0$.
A \textbf{HDA} $X$ on $L$ is a bialgebra $\ast \rightarrow X \rightarrow L$ in $\textbf{pCub}$. In other words, a HDA $X$ is a precubical set, also noted $X$, together with a specified point, the \textbf{initial state}, $i \in X_0$ and a \textbf{labelling function} $\map{\lambda}{X_1}{L}$ satisfying the equations $\lambda\circ\partial_{i,2}^0 = \lambda\circ\partial_{i,2}^1$ for $i \in \{1,2\}$ (see previous figure, right). A \textbf{morphism of HDA} from $X$ to $Y$ is a morphism $f$ of precubical sets from $X$ to $Y$ such that $f_0(i_X) = i_Y$ and $\lambda_X = \lambda_Y\circ f_1.$ HDA on $L$ and morphisms of HDA form a category that we denote by $\hda$. This category can also be defined as a the double slice category $\ast/\textbf{pCub}/L$. Remark that we are only concerned with labelling-preserving morphisms, not general morphisms as described in \cite{fahrenberg05}.

	\subsection{Original definition of partial Higher Dimensional Automata}
	\label{sub:fahr}

Originally \cite{fahrenberg15}, partial HDA are defined similarly to the concrete definition of HDA, except that the face maps can be partial functions and the local equations hold only when \emph{both} sides are well defined. There are two reasons why it fails to give the good intuition:

	\begin{wrapfigure}{r}{0.38\textwidth}
	\begin{tikzpicture}[thick,scale=3]

    \coordinate (A1) at (0, 0);
    \coordinate (A2) at (0, 1);
    \coordinate (A3) at (1, 1);
    \coordinate (A4) at (1, 0);
    \coordinate (B1) at (0.3, 0.3);
    \coordinate (B2) at (0.3, 1.3);
    \coordinate (B3) at (1.3, 1.3);
    \coordinate (B4) at (1.3, 0.3);

    \draw[dashed,color = gray, ->] (A1) -- (B1);
    \draw[dashed, ->] (B1) -- (B2);
    \draw[very thick, ->] (A2) -- (B2);
    \draw[very thick, ->] (B2) -- (B3);
    \draw[very thick, ->] (A3) -- (B3);
    \draw[very thick, ->] (A4) -- (B4);
    \draw[dashed, ->] (B1) -- (B4);
    \draw (1.28,1.23) -- (1.32,1.27);
    \draw (1.28,1.18) -- (1.32,1.22);
    \draw (1.28,1.13) -- (1.32,1.17);
    \draw (1.28,1.08) -- (1.32,1.12);
    \draw (1.28,1.03) -- (1.32,1.07);
    \draw (1.28,0.98) -- (1.32,1.02);
    \draw (1.28,0.93) -- (1.32,0.97);
    \draw (1.28,0.88) -- (1.32,0.92);
    \draw (1.28,0.83) -- (1.32,0.87);
    \draw (1.28,0.78) -- (1.32,0.82);
    \draw (1.28,0.73) -- (1.32,0.77);
    \draw (1.28,0.68) -- (1.32,0.72);
    \draw (1.28,0.63) -- (1.32,0.67);
    \draw (1.28,0.58) -- (1.32,0.62);
    \draw (1.28,0.53) -- (1.32,0.57);
    \draw (1.28,0.48) -- (1.32,0.52);
    \draw (1.28,0.43) -- (1.32,0.47);
    \draw (1.28,0.38) -- (1.32,0.42);
    \draw (1.28,0.33) -- (1.32,0.37);

    \draw[fill=gray,opacity = 0.4,pattern=north west lines] (A1) -- (B1) -- (B4) -- (A4);
    \draw[fill=gray,opacity = 0.4] (A1) -- (A2) -- (A3) -- (A4);
    \draw[fill=gray,opacity = 0.4] (A1) -- (A2) -- (B2) -- (B1);
    \draw[fill=gray,opacity = 0.4] (B1) -- (B2) -- (B3) -- (B4);
    \draw[fill=gray,opacity = 0.4] (A3) -- (B3) -- (B4) -- (A4);
    \draw[fill=gray,opacity = 0.4] (A2) -- (B2) -- (B3) -- (A3);
    
    \draw[very thick, ->] (A1) -- (A2);
    \draw[very thick, ->] (A2) -- (A3);
    \draw[very thick, ->] (A4) -- (A3);
    \draw[very thick, ->] (A1) -- (A4);
    
    \node (lf) at (0.65,0.15) {\scriptsize{not defined}};
    \node[rotate=90] (re) at (1.4,0.8) {\scriptsize{not defined}};
    
    \node at (B4) {$\bullet$};
    \node[rotate=45] at (1.35,0.25) {\scriptsize{defined in two}};
    \node[rotate=45] at (1.4,0.2) {\scriptsize{different ways}};
\end{tikzpicture}
    \vspace{-14mm}
\end{wrapfigure}
 -- first the `local' equations are not enough in the partial case. Imagine that we want to model a full cube  $c$ without its lower face, that is, $\partial_{3,3}^0$ is not defined on $c$, and such that $\partial_{1,2}^1$ is undefined on $\partial_{1,3}^1(c)$ and $\partial_{2,3}^1(c)$, that is, we remove an edge. We cannot prove using the local equations that $\partial_1^1\circ\partial_2^0\circ\partial_1^1(c) = \partial_1^1\circ\partial_2^0\circ\partial_2^1(c)$, that is, that the vertices of the cube are uniquely defined. Indeed, to prove this equality using the local equations, you can only permute two consecutive $\partial$. From $\partial_1^1\circ\partial_2^0\circ\partial_1^1(c)$, you can:\\
\indent\indent $\bullet$ either permute the first two and you obtain $\partial_1^1\circ\partial_1^1\circ\partial_3^0(c)$,\\
\indent\indent $\bullet$ or permute the last two and you obtain $\partial_1^0\circ\partial_1^1\circ\partial_1^1(c)$.\\
and both faces are not defined. On the other hand, those two should be equal because the comaps $d_1^1\circ d_2^0\circ d_1^1$ and $d_2^1\circ d_2^0\circ d_1^1$ are equal in $\Box$, and $\partial_1^1\circ\partial_2^0\circ\partial_1^1$ and $\partial_1^1\circ\partial_2^0\circ\partial_2^1$ are both defined on $c$.

		\begin{wrapfigure}{r}{0.28\textwidth}
	\vspace{-7mm}
	\begin{tikzpicture}[scale=1]
	
	\node (seg) at (1,0.5) {\textbf{segment}};
	\node (sseg) at (1,-1.5) {\textbf{split segment}}	;
	\node (11) at (0,0) {$\bullet$};
	\node (12) at (2,0) {$\bullet$};
	\node (21) at (-0.5,-1) {$\bullet$};
	\node (22) at (2.5,-1) {$\bullet$};
	\node (21') at (0,-1) {};
	\node (22') at (2,-1) {};
	
	\path[->,font=\scriptsize]
		(11) edge node[above]{$a$} (12)
		(21') edge node[below]{$a$} (22');
	\path[->, dotted]
		(11) edge (21)
		(12) edge (22)
		(1,-0.2) edge (1,-0.7);
			
	\end{tikzpicture}
	\vspace{-10mm}
	\end{wrapfigure}
-- secondly, the notion of morphism is not good (or at least, ambiguous). The equations $f_n\circ\partial_{i,n,X}^\alpha = \partial_{i,n,Y}^\alpha\circ f_{n+1}$ hold in \cite{fahrenberg15} only when \emph{both} face maps are defined, which authorises many morphisms. For example, consider the segment $I$, and the `split' segment $I'$ which is defined as $I$, except that no face maps are defined (geometrically, this corresponds to two points and an open segment). The identity map from $I$ to $I'$ is a morphism of partial precubical sets in the sense of \cite{fahrenberg15}, which is unexpected. A bad consequence of that is that the notion of paths in a partial HDA does not correspond to morphisms from some particular partial HDA, and paths are not preserved by morphisms, as we will see later.

	\subsection{Partial Higher Dimensional Automata as lax functors}
	\label{sub:phda}

The idea is to generalise the `presheaf' definition of precubical sets. The problem is to deal with partial functions and when two of them should coincide. Let \textbf{pSet} be the category of sets and partial functions. A partial function $\map{f}{X}{Y}$ can be either seen as a pair $(A,f)$ of a subset $A \subseteq X$ and a total function $\map{f}{A}{Y}$, or as a functional relation $f \subseteq X\times Y$, that is, a relation such that for every $x\in X$, there is at most one $y \in Y$ with $(x,y) \in f$. We will freely use both views in the following. For two partial maps $\map{f,g}{X}{Y}$, we denote by $f \equiv g$ if and only if for every $x \in X$ such that $f(x)$ and $g(x)$ are defined, then $f(x) = g(x)$. Note that this is not equality, but equality on the intersection of the domains. We also write $f \subseteq g$ if and only if $f$ is include in $g$ as a relation, that is, if and only if, for every $x \in X$ such that $f(x)$ is defined, then $g(x)$ is defined and $f(x) = g(x)$. By a \textbf{lax functor} $\pmap{F}{\C}$, we mean the following data \cite{niefield10}:\\
\indent -- for every object $c$ of $\C$, a set $Fc$,\\
\indent -- for every morphism $\map{i}{c}{c'}$, a partial function $\map{Fi}{Fc}{Fc'}$\\
satisfying that $F\text{id}_c = \text{id}_{Fc}$ and $Fj \circ Fi \subseteq F(j\circ i)$.\\
The point is that partial precubical sets as defined in \cite{fahrenberg15} do not satisfy the second condition, while they should. In addition, this definition will authorise a square to have vertices, that is, that some $\partial\partial$ are defined, while having no edge, that is, no $\partial$ defined. This may be useful to define paths as discrete traces in \cite{fajstrup05} (that we will call \emph{shortcuts} later), that is, paths that can go directly from a point to a square for example. Observe also that if $j \circ i = j'\circ i'$ then $Fj \circ Fi \equiv Fj' \circ Fi'$, which gives us the local equations from \cite{fahrenberg15}. A \textbf{partial precubical set} $X$ is then a lax functor $\pmap{F}{\Box^{op}}$. It becomes harder to describe explicitly what a partial precubical set is, since we cannot restrict to the $\partial_i^\alpha$ anymore. It is a collection of sets $(X_n)_{n\in\nat}$ together with a collection of \emph{partial} functions $(\map{\partial_{i_1 < \ldots < i_k}^{\alpha_1, \ldots, \alpha_k}}{X_{n+k}}{X_n})$ satisfying the inclusions 
$\partial_{j_1 < \ldots < j_m}^{\beta_1, \ldots, \beta_m}\circ\partial_{i_1 < \ldots < i_n}^{\alpha_1, \ldots, \alpha_n} \subseteq \partial_{k_1 < \ldots < k_{n+m}}^{\gamma_1, \ldots, \gamma_{n+m}}$
where the $k_s$ and $\gamma_s$ are defined as follows. $(k_1< \ldots < k_{n+m} ; \gamma_1, \ldots, \gamma_{n+m}) = (i_1<\ldots < i_n ; \alpha_1, \ldots, \alpha_n) \star (j_1< \ldots< j_m ; \beta_1, \ldots, \beta_m)$ where $\star$ is defined by induction on $n+m$:\\
\indent --  if $n=0$, $\epsilon \star (j_1< \ldots< j_m ; \beta_1, \ldots, \beta_m) = (j_1< \ldots< j_m ; \beta_1, \ldots, \beta_m)$,\\
\indent --  if $m = 0$, $(i_1<\ldots< i_n ; \alpha_1, \ldots, \alpha_n) \star \epsilon = (i_1<\ldots< i_n ; \alpha_1, \ldots, \alpha_n)$,\\
\indent --  if $i_1 \leq j_1$, $(i_1<\ldots< i_n ; \alpha_1, \ldots, \alpha_n) \star (j_1< \ldots< j_m ; \beta_1, \ldots, \beta_m) = (i_1;\alpha_1).((i_2<\ldots< i_n ; \alpha_2, \ldots, \alpha_n) \star (j_1+1< \ldots< j_m+1 ; \beta_1, \ldots, \beta_m))$,\\
\indent --  if $i_1 > j_1$, $(i_1<\ldots< i_n ; \alpha_1, \ldots, \alpha_n) \star (j_1< \ldots< j_m ; \beta_1, \ldots, \beta_m) = (j_1;\beta_1).((i_1<\ldots< i_n ; \alpha_1, \ldots, \alpha_n) \star (j_2< \ldots< j_m ; \beta_2, \ldots, \beta_m))$.\\
A \textbf{function-valued op-lax transformation} \cite{niefield10} from $\pmap{F}{\C}$ to $\pmap{G}{\C}$ is a collection $(f_c)_{c\in Ob(\C)}$ of \emph{total} functions such that for every $\map{i}{c}{c'}$, $f_{c'}\circ F(i) \subseteq G(i)\circ f_{c}$. A \textbf{morphism of partial precubical sets} from $X$ to $Y$ is then a function-valued op-lax transformation. In other words, this is a collection of \emph{total} functions $(\map{f_n}{X_n}{Y_n})_{n\in\nat}$ satisfying the equations $f_n\circ\partial_{i_1 < \ldots < i_k}^{\alpha_1, \dots, \alpha_k} \subseteq \partial_{i_1 < \ldots < i_k}^{\alpha_1, \ldots, \alpha_k}\circ f_{n+k}$. Partial precubical sets and morphisms of partial precubical sets form a category that we denote by $\textbf{ppCub}$. $\textbf{pCub}$ is a full subcategory of $\textbf{ppCub}$. In particular, the precubical sets $\ast$ and $L$ are partial precubical sets. A \textbf{partial HDA} $X$ on $L$ is a partial precubical set, also noted $X$, together with a specified point, the \textbf{initial state} $i \in X_0$ and a morphism of ppCub, the \textbf{labelling functions}, $(\map{\lambda_n}{X_n}{L^n})_{n\in\nat}$. A \textbf{morphism of pHDA} from $X$ to $Y$ is a morphism $f$ of partial precubical sets from $X$ to $Y$ such that $f_0(i_X) = i_Y$ and $\lambda_X = \lambda_Y\circ f$. Partial HDA on $L$ and morphisms of partial HDA form a category that we note $\phda$. In other words, this is the double slice category $\ast/\textbf{ppCub}/L$.

	\subsection{Completion of a pHDA}
	\label{sub:refl}

%

Let us describe how it is possible to construct a HDA from a pHDA $X$, by `completing' $X$, that is, by adding the faces that are missing, and by connecting the faces that are not. Let $$Y_n = \{((i_1 < \ldots < i_k ; \alpha_1, \ldots, \alpha_k),x) \mid x \in X_{n+k} \wedge i_k \leq n+k\}$$
$Y = (Y_n)_{n\in \nat}$ is intuitively the collection of all abstract faces of all cubes of $X$, that is, pairs of a cube and all possible ways to define a face from it. Of course, some of those are the same, since there are several ways to describe a cube as the face of some other cube. Define $\sim$ as the smallest equivalence relation such that:\\
\indent --  if $\partial_{i_1 < \ldots < i_k}^{\alpha_1, \ldots, \alpha_k}(x)$ is defined, then $((i_1 < \ldots < i_k ; \alpha_1, \ldots, \alpha_k), x) \sim (\epsilon, \partial_{i_1 < \ldots < i_k}^{\alpha_1, \ldots, \alpha_k}(x))$. This means that, if a face of a cube exists in $X$, this face is identified with both abstract faces $(\epsilon, \partial_{i_1 < \ldots < i_k}^{\alpha_1, \ldots, \alpha_k}(x))$ (i.e., the cube $\partial_{i_1 < \ldots < i_k}^{\alpha_1, \ldots, \alpha_k}(x)$ itself) and $((i_1 < \ldots < i_k ; \alpha_1, \ldots, \alpha_k), x)$ (i.e., the face of $x$, which consists of taking the $(i_k,\alpha_k)$ face, then the $(i_{k-1},\alpha_{k-1})$ face, and so on).\\
\indent --  if $((i_1 < \ldots < i_k ; \alpha_1, \ldots, \alpha_k), x) \sim ((j_1 < \ldots < j_l ; \beta_1, \ldots, \beta_l), y)$, then $((i_1 < \ldots < i_k ; \alpha_1, \ldots, \alpha_k)\star(i,\alpha), x) \sim ((j_1 < \ldots < j_l ; \beta_1, \ldots, \beta_l)\star(i,\alpha), y)$. This means that if two abstract faces coincide, then taking both their $(i,\alpha)$ face gives two abstract faces that also coincide.\\
Let $\chi(X)_n = Y_n/\sim$ and we denote by $\ll (i_1 < \ldots < i_k ; \alpha_1, \ldots, \alpha_k), x \gg$ the equivalence class of $((i_1 < \ldots < i_k ; \alpha_1, \ldots, \alpha_k), x)$ modulo $\sim$. We define the $i$-face map as $\partial_i^\alpha(\ll(i_1 < \ldots < i_k ; \alpha_1, \ldots, \alpha_k), x\gg) = ~ \ll(i_1 < \ldots < i_k ; \alpha_1, \ldots, \alpha_k)\star(i,\alpha), x\gg$, the initial state as $\ll\epsilon,i\gg$ and the labelling function as $\lambda(\ll (i_1 < \ldots < i_k ; \alpha_1, \ldots, \alpha_k), x \gg) = \delta_{i_1}^{\alpha_1}\circ\ldots\circ\delta_{i_k}^{\alpha_k}(\lambda(x))$. 

\begin{theorem}
\label{the:reflection}
$\chi$ is a well-defined functor and is the left adjoint of $\tau$, the injection of $\hda$ into $\phda$. Furthermore, $\hda$ is a reflective subcategory of $\phda$.
\end{theorem}

Now, we can define the \textbf{geometric realisation} of a pHDA $X$ as the subspace of the realisation of $\chi(X)$ consisting of points whose carrier is of the form $\ll \epsilon, x \gg$ for some $x \in X$. This really corresponds to the drawings we have been using to depict pHDA until now.

\section{Paths in partial Higher Dimensional Automata}
\label{sec:path}

Executions of HDA are defined using the notion of paths. Those paths describe the succession of starting and finishing of actions in a HDA. For example, a HDA can start an action then start another at the same time, and finish the two actions. This sequence is then not just a sequence of 1-dimensional transitions, since some actions can be made at the same time, but a sequence of hypercubes corresponding to the evolution of the state of the system. We will formalise this idea in Section \ref{sub:paths}, and we will see in particular that those paths can be encoded in the category $\phda$ (while it is not possible in the category $\hda$) as morphisms from particular pHDA, called path shapes. In Section \ref{sub:open}, let us first start by recalling the general framework of \cite{joyal96}.

	\subsection{Path category, open maps, coverings}
	\label{sub:open}
	
	In the general framework of \cite{joyal96}, we start with a category $\M$ of systems, together with a subcategory $\P$ of execution shapes. For example, keep in mind the case where $\M$ is the category of transition systems and $\P$ is the full subcategory of finite linear systems. One interesting remark about this case is that executions of a given systems are in bijective correspondance with morphisms from a finite linear system to this given system. This means that to reason about behaviours of such systems, it is enough to reason about morphisms and execution shapes.
	\begin{wrapfigure}{r}{0.25\textwidth}
	\vspace{-9mm}
	\begin{tikzpicture}[scale=0.8]
		
	\node (B') at (0,0) {$Y'$};
	\node (B) at (0,3) {$X'$};
	\node (S) at (3,0) {$Y$};
	\node (T) at (3,3) {$X$};

	\path[->,font=\scriptsize]
		(B) edge node[left]{$g$} (B')
		(B') edge node[below]{$y$} (S)
		(T) edge node[right]{$f$} (S)
		(B) edge node[above]{$x$} (T);
		
	\path[->, font = \scriptsize, dotted]
		(B') edge node[above]{$\theta$} (T);
			
\end{tikzpicture}
	\vspace{-14mm}
\end{wrapfigure}
	This idea was formalised by describing precisely which morphisms are witnesses for the existence of a bisimulation between systems. This description uses right lifting properties: we say that a morphism $\map{f}{X}{Y}$ has the \textbf{right lifting property with respect to $\map{g}{X'}{Y'}$} if for every $\map{x}{X'}{X}$ and $\map{y}{Y'}{Y}$ such that $f\circ x = y \circ g$, there exists $\map{\theta}{Y'}{X}$ such that $x = \theta\circ g$ and $f\circ\theta = y$. For example, let us assume that $f$ is a morphism of transition systems and that $X'$ and $Y'$ are finite linear systems. Then $x$ (resp. $y$) is the same as an execution in $X$ (resp. $Y$), and $f\circ x = y \circ g$ means that the execution $y$ is a extension of the image of the execution $x$ by $f$. The right lifting property means that the longer execution $y$ of $Y$ can be lifted to a longer execution $\theta$ of $X$, that is, $\theta$ is an extension of $x$ and the image of $\theta$ by $f$ is $y$. This property of lifting longer executions is precisely the property needed on a morphism to make its graph relation a bisimulation. They are also very similar to morphisms of coalgebras \cite{jacobs16}. We call \textbf{$\P$-open} (or simply open when $\P$ is clear), a morphism that has the right lifting property with respect to every morphism in $\P$. From open maps, it is possible to describe similarity and bismilarity as the existence of a span of morphisms/open maps, and many kinds of bisimilarities can be captured in this way \cite{joyal96}. An open map is said to be a \textbf{$\P$-covering} (or simply covering) if furthermore the lifts in the right lifting properties are unique. Being a covering is a very strong requirement, as they correspond to partial unfolding of a system.

	\subsection{Encoding paths in pHDA}
	\label{sub:paths}

In this section, we describe the classical notion of execution of HDA from \cite{vanglabbeek05}, extended to partial HDA in \cite{fahrenberg15}, defined using the notion of path. We then show 
	\begin{wrapfigure}{r}{0.28\textwidth}
		\vspace{-10mm}
\begin{tikzpicture}[scale=0.8]
	\node (11p) at (5,0) {};
	\node (12p) at (5,3) {};
	\node (21p) at (8,0) {};
	\node (22p) at (8,3) {};
	\node[rotate=45] (ul2) at (4.6,3.4) {\scriptsize{not}};
	\node[rotate=45] (ul) at (4.8,3.2) {\scriptsize{defined}};
	\node (ue) at (6.5,3.3) {\scriptsize{not defined}};
	\draw[draw = white, fill = gray!50] (5,0) rectangle (8,3);
	\path[->,thick,font=\scriptsize]
		(11p) edge (12p)
		(21p) edge (22p)
		(11p) edge (21p);
	\path[very thick, red]
		(5,0) edge (6.5,0)
		(6.5,0) edge (6.5,1.5)
		(6.5,1.5) edge (8,1.5);
	\draw (5.1,2.9) -- (5.3,3.1);
	\draw (5.3,2.9) -- (5.5,3.1);
	\draw (5.5,2.9) -- (5.7,3.1);
	\draw (5.7,2.9) -- (5.9,3.1);
	\draw (5.9,2.9) -- (6.1,3.1);
	\draw (6.1,2.9) -- (6.3,3.1);
	\draw (6.3,2.9) -- (6.5,3.1);
	\draw (6.5,2.9) -- (6.7,3.1);
	\draw (6.7,2.9) -- (6.9,3.1);
	\draw (6.9,2.9) -- (7.1,3.1);
	\draw (7.1,2.9) -- (7.3,3.1);
	\draw (7.3,2.9) -- (7.5,3.1);
	\draw (7.5,2.9) -- (7.7,3.1);
	\draw (7.7,2.9) -- (7.9,3.1);
	\node at (5,0) {$\bullet$};
	\node at (5,3) {$\times$};
	\node at (8,0) {$\bullet$};
	\node at (8,3) {$\bullet$};
	\node (0) at (4.8,-0.2) {\scriptsize{$0$}};
	\node (1) at (8.2,-0.2) {\scriptsize{$1$}};
	\node (2) at (8.2,3.2) {\scriptsize{$2$}};
	\node (al) at (4.8,1.5) {\scriptsize{$\alpha$}};
	\node (be) at (6.5,-0.2) {\scriptsize{$\beta$}};
	\node (ga) at (8.2,1.5) {\scriptsize{$\gamma$}};
	\node (c) at (6.5,1.7) {\scriptsize{$c$}};
	\node (lab) at (6.5,-0.55) {\scriptsize{In red: path}};
	\node (lab2) at (6.5, -1) {\scriptsize{$\pi = 0 \xrightarrow{1,0} \beta \xrightarrow{2,0} c \xrightarrow{1,1} \gamma$}};
	\node (lab3) at (6.5,-1.45) {\scriptsize{in the pHDA $X$}};
\end{tikzpicture}
\vspace{-14mm}
\end{wrapfigure}
that those executions can be encoded as an execution shapes subcategory, as in the general framework of \cite{joyal96}, proving in particular that paths are in bijective correspondance with a class of morphisms. A \textbf{path} $\pi$ of a HDA $X$ is a sequence $i = x_0 \xrightarrow{j_1,\alpha_1} x_1 \xrightarrow{j_2,\alpha_2} \ldots \xrightarrow{j_n,\alpha_n} x_n$ where $x_k \in X$, $j_k > 0$ and $\alpha_k \in \{0,1\}$ are such that for every $k$:\\
\indent --  if $\alpha_k =0$, then $x_{k-1} = \partial_{j_k}^{0}(x_k)$,\\
\indent --  if $\alpha_k =1$, then $x_{k} = \partial_{j_k}^{1}(x_{k-1})$.\\
This definition can easily be extended to pHDA, by requiring that the $j_k$-face maps are defined on $x_k$ or $x_{k-1}$. A natural property of executions and morphisms is that morphisms map executions to executions. This is the case here (while it is not for \cite{fahrenberg15}, e.g., the split segment):

\begin{proposition}
\label{prop:morph}
If $\map{f}{X}{Y}$ is a map of pHDA and if $\pi =  x_0 \xrightarrow{j_1,\alpha_1} x_1 \xrightarrow{j_2,\alpha_2} \ldots \xrightarrow{j_n,\alpha_n} x_n$ is a path in $X$, then $\pi' =  f(x_0) \xrightarrow{j_1,\alpha_1} f(x_1) \xrightarrow{j_2,\alpha_2} \ldots \xrightarrow{j_n,\alpha_n} f(x_n)$ is a path in $Y$.
\end{proposition}

One advantage of considering pHDA instead of HDA is that paths can be encoded in pHDA, which is not really possible in HDA. It is done as follows. A \textbf{spine} $\sigma$ is a sequence $(0,\epsilon) = (d_0,w_0) \xrightarrow{j_1,\alpha_1} (d_1,w_1) \xrightarrow{j_2,\alpha_2} \ldots \xrightarrow{j_n,\alpha_n} (d_n,w_n)$ where $j_k > 0$, $d_k \in \nat$, $w_k \in L^{d_k}$ and $\alpha_k \in \{0,1\}$ are such that:\\
\indent --  if $\alpha_k =0$, then $d_{k-1} = d_k - 1$, $\delta_{j_k}(w_k) = w_{k-1}$ and $j_k \leq d_k$,\\
\indent --  if $\alpha_k =1$, then $d_{k} = d_{k-1} - 1$, $\delta_{j_k}(w_{k-1}) = w_k$ and $j_k \leq d_{k-1}$.\\
\vspace{-5mm}
\begin{wrapfigure}{r}{0.3\textwidth}
		\vspace{-8mm}
\begin{tikzpicture}[scale=0.8]
	\node (11p) at (5,0) {};
	\node (12p) at (5,3) {};
	\node (21p) at (8,0) {};
	\node (22p) at (8,3) {};
	\node[rotate=45] (ul2) at (4.6,3.4) {\scriptsize{not}};
	\node[rotate=45] (ul) at (4.8,3.2) {\scriptsize{defined}};
	\node[rotate=-45] (ul2) at (8.4,3.4) {\scriptsize{not}};
	\node[rotate=-45] (ul) at (8.2,3.2) {\scriptsize{defined}};
	\node[rotate=45] (ul2) at (8.2,-0.2) {\scriptsize{not}};
	\node[rotate=45] (ul) at (8.4,-0.4) {\scriptsize{defined}};
	\node (ue) at (6.5,3.3) {\scriptsize{not defined}};
	\node[rotate=90] (ue) at (4.7,1.5) {\scriptsize{not defined}};
	\draw[draw = white, fill = gray!50] (5,0) rectangle (8,3);
	\path[->,thick,font=\scriptsize]
		(21p) edge (22p)
		(11p) edge (21p);
	\draw (5.1,2.9) -- (5.3,3.1);
	\draw (5.3,2.9) -- (5.5,3.1);
	\draw (5.5,2.9) -- (5.7,3.1);
	\draw (5.7,2.9) -- (5.9,3.1);
	\draw (5.9,2.9) -- (6.1,3.1);
	\draw (6.1,2.9) -- (6.3,3.1);
	\draw (6.3,2.9) -- (6.5,3.1);
	\draw (6.5,2.9) -- (6.7,3.1);
	\draw (6.7,2.9) -- (6.9,3.1);
	\draw (6.9,2.9) -- (7.1,3.1);
	\draw (7.1,2.9) -- (7.3,3.1);
	\draw (7.3,2.9) -- (7.5,3.1);
	\draw (7.5,2.9) -- (7.7,3.1);
	\draw (7.7,2.9) -- (7.9,3.1);
	\draw (4.9,0.1) -- (5.1,0.3);
	\draw (4.9,0.3) -- (5.1,0.5);
	\draw (4.9,0.5) -- (5.1,0.7);
	\draw (4.9,0.7) -- (5.1,0.9);
	\draw (4.9,0.9) -- (5.1,1.1);
	\draw (4.9,1.1) -- (5.1,1.3);
	\draw (4.9,1.3) -- (5.1,1.5);
	\draw (4.9,1.5) -- (5.1,1.7);
	\draw (4.9,1.7) -- (5.1,1.9);
	\draw (4.9,1.9) -- (5.1,2.1);
	\draw (4.9,2.1) -- (5.1,2.3);
	\draw (4.9,2.3) -- (5.1,2.5);
	\draw (4.9,2.5) -- (5.1,2.7);
	\draw (4.9,2.7) -- (5.1,2.9);
	\node at (5,0) {$\bullet$};
	\node at (5,3) {$\times$};
	\node at (8,0) {$\times$};
	\node at (8,3) {$\times$};
	\node at (4.8,-0.2) {\scriptsize{$\epsilon$}};
	\node at (6.5,-0.2) {\scriptsize{$a$}};
	\node at (6.5,1.5) {\scriptsize{$ab$}};
	\node at (8.2,1.5) {\scriptsize{$b$}};
	\node (lab2) at (6.5, -1) {\scriptsize{path shape of the spine}};
	\node (lab3) at (6.5,-1.45) {\scriptsize{$\sigma = (0,\epsilon) \xrightarrow{1,0} (1,a) \xrightarrow{2,0} (2,ab)$}};
	\node (lab3) at (6.5,-1.9) {\scriptsize{$~~~~~~~~~~~~~~~~~~~~~~~~~~\xrightarrow{1,1} (1,b)$}};
\end{tikzpicture}
\vspace{-20mm}
\end{wrapfigure}
A path $\pi$ has a underlying spine $\sigma_\pi$ by mapping $x_k$ to the pair of its dimension and its label. A spine $\sigma$ induces a pHDA $B\sigma$ as follows:
\vspace{-3mm}
\begin{itemize}
	\item $B\sigma_p = \{k \in \{0,\ldots,n\} \mid d_k = p\}$,
	\item the partial face maps $\partial_{i_1 < \ldots < i_n}^{\alpha_1, \ldots, \alpha_n}$ are the smallest (as relations ordered by inclusion) partial functions such that:
		\begin{itemize}
			\item if $\alpha_k = 0$, then $\partial_{j_k}^{0}(k) = k-1$,
			\item if $\alpha_k = 1$, then $\partial_{j_k}^{1}(k-1) = k$,
			\item $\partial_{j_1 < \ldots < j_m}^{\beta_1, \ldots, \beta_m}\circ\partial_{i_1 < \ldots < i_n}^{\alpha_1, \ldots, \alpha_n} \subseteq \partial_{k_1 < \ldots < k_{n+m}}^{\gamma_1, \ldots, \gamma_{n+m}}$, for\\ $(k_1, \ldots, k_{n+m} ; \gamma_1, \ldots, \gamma_{n+m}) =$\\ $(i_1,\ldots, i_n ; \alpha_1, \ldots, \alpha_n) \star (j_1, \ldots, j_m ; \beta_1, \ldots, \beta_m)$.
		\end{itemize}
	\item the initial state is $0$,
	\item the labelling functions $\lambda_n$ map $k$ to $w_k$.
\end{itemize}
\vspace{-3mm}
By a \textbf{path shape}, we mean such a pHDA $B\sigma$. The set $\spine$ of spines can be partially ordered by prefix. $B$ can then be extended to an embedding from $\spine$ to $\phda$. We note $\ps$ the image of this embedding, i.e., the full sub-category of path shapes.

\begin{proposition}
\label{prop:bij}
There is a bijection between paths in a pHDA $X$ and morphisms of pHDA from a path shape to $X$.
\end{proposition}

Again, this is not true with the definition of morphisms from \cite{fahrenberg15} (e.g., the split segment). As an example, the red path $\pi$ above corresponds to a morphism from the path shape $B\sigma$ to $X$.

\section{Trees and unfolding in pHDA}
\label{sec:unfo}

In this section, we introduce our notion of trees. Following \cite{dubut16}, we consider trees as colimits (or glueings of paths). Section \ref{sub:colimits} is dedicated to proving that those colimits actually exist, by giving an explicit construction of those. From this explicit construction, we will describe the kind of unique path properties that are satisfied by those trees in Section \ref{sub:unique}. Starting by showing, that the strict unicity of path fails, we then describe a notion of homotopy, the confluent homotopy, which is weaker than the one from \cite{vanglabbeek05}, for which every tree has the property that there is exactly one homotopy class of paths form the initial state to any state. We will also see that, because the face maps of trees are defined in a local way, they do not have any shortcuts, that is, paths that `skip' dimensions, for example, going from a point to a square without going through a segment. Finally, in Section \ref{sub:unfolding}, we will prove that those two properties -- the unicity of paths modulo confluent homotopy, and the non-existence of shortcuts -- completely characterise trees. This proof will use a suitable notion of unfolding of pHDA, showing furthermore that trees form a coreflective subcategory of pHDA.


\subsection{Trees, as colimits of paths in pHDA}
	\label{sub:colimits}

	In this section, we give an explicit construction of colimits of diagrams with values in path shapes. Those will be our first definition of trees in pHDA, following \cite{dubut16}. Let $\map{D}{\C}{\ps}$ be a small diagram with values in $\ps$, that is, a functor from $\C$ to $\ps$. Let us use some notations: for every object $u$ of $\C$, 
$Du = B\sigma_u$ with $\sigma_u = (d_0^u,w_0^u) \xrightarrow{j_1^u,\alpha_1^u} (d_1^u,w_1^u) \xrightarrow{j_2^u,\alpha_2^u} \ldots \xrightarrow{j_{l_u}^u,\alpha_{l_u}^u} (d^u_{l_u},w^u_{l_u})$. The definition of the colimit $\col{D}$ will be in two steps. 
The first step consists in putting all the paths $Du$ side-by-side, and in glueing them together, along the morphisms $Df$, for every morphism $f$ of $\C$. This is done as follows. Define $(X_n)_{n\in\nat}$ to be:\\
\indent --  $X_0 = \{(u,k) \mid u \in \C, k \leq l_u \wedge d^u_k = 0\} \sqcup\{\epsilon\}$,\\
\indent --  $X_n = \{(u,k) \mid u \in \C, k \leq l_u \wedge d^u_k = n\}$.\\
We quotient $X_n$ by the smallest equivalence relation $\sim$ (for inclusion) such that:\\
\indent --  for every $u$, $(u,0) \sim \epsilon$,\\
\indent --  if $\map{i}{u}{v} \in \C$, and if $k \leq l_u, l_v$, then $(u,k) \sim (v,k)$.\\
We denote by $Y_n$ the quotient $X_n/\sim$, and by $\eqclp{u,k}$ the equivalence class of $(u,k)$ modulo $\sim$.\\
At this stage, we still do not have the colimit because it is not possible to define the face maps. Let us consider the following example.
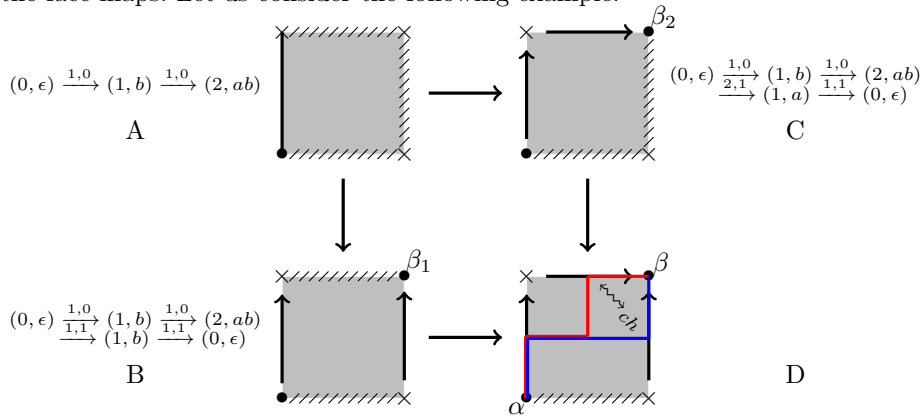
\vspace{-8mm}
\begin{center}
\begin{tikzpicture}[scale=0.65]
	\draw[draw = white, fill = gray!50] (0,0) rectangle (2.5,2.5);
	\node (00) at (-3,1.7) {\scriptsize{$(0,\epsilon) \xrightarrow{1,0} (1,b) \xrightarrow{1,0} (2,ab)$}};
	\node (00p) at (-2.65,1.3) {\scriptsize{$\xrightarrow{1,1} (1,b) \xrightarrow{1,1} (0,\epsilon)$}};
	\node (00pp) at (-3,0.5) {B};
	\node (0000) at (0,0) {$\bullet$};
	\node (0001) at (0,2.5) {$\times$};
	\node (0010) at (2.5,0) {$\times$};
	\node (0011) at (2.5,2.5) {$\bullet$};
	\node (0011p) at (2.8,2.8) {$\beta_1$};
	\draw[very thick, ->] (0000) to (0001);
	\draw[very thick, ->] (0010) to (0011);
	\draw (0.15,2.4) -- (0.35,2.6);
	\draw (0.35,2.4) -- (0.55,2.6);
	\draw (0.55,2.4) -- (0.75,2.6);
	\draw (0.75,2.4) -- (0.95,2.6);
	\draw (0.95,2.4) -- (1.15,2.6);
	\draw (1.15,2.4) -- (1.35,2.6);
	\draw (1.35,2.4) -- (1.55,2.6);
	\draw (1.55,2.4) -- (1.75,2.6);
	\draw (1.75,2.4) -- (1.95,2.6);
	\draw (1.95,2.4) -- (2.15,2.6);
	\draw (2.15,2.4) -- (2.35,2.6);
	
	\draw (0.15,-0.1) -- (0.35,0.1);
	\draw (0.35,-0.1) -- (0.55,0.1);
	\draw (0.55,-0.1) -- (0.75,0.1);
	\draw (0.75,-0.1) -- (0.95,0.1);
	\draw (0.95,-0.1) -- (1.15,0.1);
	\draw (1.15,-0.1) -- (1.35,0.1);
	\draw (1.35,-0.1) -- (1.55,0.1);
	\draw (1.55,-0.1) -- (1.75,0.1);
	\draw (1.75,-0.1) -- (1.95,0.1);
	\draw (1.95,-0.1) -- (2.15,0.1);
	\draw (2.15,-0.1) -- (2.35,0.1);
	
	\draw[draw = white, fill = gray!50] (0,5) rectangle (2.5,7.5);
	\node (01) at (-3,6.5) {\scriptsize{$(0,\epsilon) \xrightarrow{1,0} (1,b) \xrightarrow{1,0} (2,ab)$}};
	\node (01p) at (-3,5.5) {A};
	\node (0100) at (0,5) {$\bullet$};
	\node (0101) at (0,7.5) {$\times$};
	\node (0110) at (2.5,5) {$\times$};
	\node (0111) at (2.5,7.5) {$\times$};
	\draw[very thick] (0,5) to (0,7.5);
	\draw (0.15,4.9) -- (0.35,5.1);
	\draw (0.35,4.9) -- (0.55,5.1);
	\draw (0.55,4.9) -- (0.75,5.1);
	\draw (0.75,4.9) -- (0.95,5.1);
	\draw (0.95,4.9) -- (1.15,5.1);
	\draw (1.15,4.9) -- (1.35,5.1);
	\draw (1.35,4.9) -- (1.55,5.1);
	\draw (1.55,4.9) -- (1.75,5.1);
	\draw (1.75,4.9) -- (1.95,5.1);
	\draw (1.95,4.9) -- (2.15,5.1);
	\draw (2.15,4.9) -- (2.35,5.1);
	
	\draw (0.15,7.4) -- (0.35,7.6);
	\draw (0.35,7.4) -- (0.55,7.6);
	\draw (0.55,7.4) -- (0.75,7.6);
	\draw (0.75,7.4) -- (0.95,7.6);
	\draw (0.95,7.4) -- (1.15,7.6);
	\draw (1.15,7.4) -- (1.35,7.6);
	\draw (1.35,7.4) -- (1.55,7.6);
	\draw (1.55,7.4) -- (1.75,7.6);
	\draw (1.75,7.4) -- (1.95,7.6);
	\draw (1.95,7.4) -- (2.15,7.6);
	\draw (2.15,7.4) -- (2.35,7.6);
	
	\draw (2.4,5.15) -- (2.6,5.35);
	\draw (2.4,5.35) -- (2.6,5.55);
	\draw (2.4,5.55) -- (2.6,5.75);
	\draw (2.4,5.75) -- (2.6,5.95);
	\draw (2.4,5.95) -- (2.6,6.15);
	\draw (2.4,6.15) -- (2.6,6.35);
	\draw (2.4,6.35) -- (2.6,6.55);
	\draw (2.4,6.55) -- (2.6,6.75);
	\draw (2.4,6.75) -- (2.6,6.95);
	\draw (2.4,6.95) -- (2.6,7.15);
	\draw (2.4,7.15) -- (2.6,7.35);
	
	\draw[draw = white, fill = gray!50] (5,5) rectangle (7.5,7.5);
	\node (11) at (10.5,6.7) {\scriptsize{$(0,\epsilon) \xrightarrow{1,0} (1,b) \xrightarrow{1,0} (2,ab)$}};
	\node (11p) at (10.85,6.3) {\scriptsize{$\xrightarrow{2,1} (1,a) \xrightarrow{1,1} (0,\epsilon)$}};
	\node (11pp) at (10.5,5.5) {C};
	\node (1100) at (5,5) {$\bullet$};
	\node (1101) at (5,7.5) {$\times$};
	\node (1110) at (7.5,5) {$\times$};
	\node (1111) at (7.5,7.5) {$\bullet$};
	\node (1111p) at (7.8,7.8) {$\beta_2$};
	\draw[very thick,->] (1100) to (1101);
	\draw[very thick,->] (1101) to (1111);
	\draw (5.15,4.9) -- (5.35,5.1);
	\draw (5.35,4.9) -- (5.55,5.1);
	\draw (5.55,4.9) -- (5.75,5.1);
	\draw (5.75,4.9) -- (5.95,5.1);
	\draw (5.95,4.9) -- (6.15,5.1);
	\draw (6.15,4.9) -- (6.35,5.1);
	\draw (6.35,4.9) -- (6.55,5.1);
	\draw (6.55,4.9) -- (6.75,5.1);
	\draw (6.75,4.9) -- (6.95,5.1);
	\draw (6.95,4.9) -- (7.15,5.1);
	\draw (7.15,4.9) -- (7.35,5.1);
	
	\draw (7.4,5.15) -- (7.6,5.35);
	\draw (7.4,5.35) -- (7.6,5.55);
	\draw (7.4,5.55) -- (7.6,5.75);
	\draw (7.4,5.75) -- (7.6,5.95);
	\draw (7.4,5.95) -- (7.6,6.15);
	\draw (7.4,6.15) -- (7.6,6.35);
	\draw (7.4,6.35) -- (7.6,6.55);
	\draw (7.4,6.55) -- (7.6,6.75);
	\draw (7.4,6.75) -- (7.6,6.95);
	\draw (7.4,6.95) -- (7.6,7.15);
	\draw (7.4,7.15) -- (7.6,7.35);
	
	\draw[draw = white, fill = gray!50] (5,0) rectangle (7.5,2.5);
	\node (1000) at (5,0) {$\bullet$};
	\node (1001) at (5,2.5) {$\times$};
	\node (1010) at (7.5,0) {$\times$};
	\node (1011) at (7.5,2.5) {$\bullet$};
	\node (1011p) at (10.5,0.5) {D};
	\node (al) at (4.8,-0.2) {$\alpha$};
	\node (be) at (7.75,2.75) {$\beta$};
	\draw[very thick,->] (1000) to (1001);
	\draw[very thick,->] (1001) to (1011);
	\draw[very thick, ->] (1010) to (1011);
	\draw (5.15,-0.1) -- (5.35,0.1);
	\draw (5.35,-0.1) -- (5.55,0.1);
	\draw (5.55,-0.1) -- (5.75,0.1);
	\draw (5.75,-0.1) -- (5.95,0.1);
	\draw (5.95,-0.1) -- (6.15,0.1);
	\draw (6.15,-0.1) -- (6.35,0.1);
	\draw (6.35,-0.1) -- (6.55,0.1);
	\draw (6.55,-0.1) -- (6.75,0.1);
	\draw (6.75,-0.1) -- (6.95,0.1);
	\draw (6.95,-0.1) -- (7.15,0.1);
	\draw (7.15,-0.1) -- (7.35,0.1);
	
	\draw[->, very thick] (3,1.25) to (4.5,1.25);
	\draw[->, very thick] (3,6.25) to (4.5,6.25);
	\draw[<-, very thick] (1.25,3) to (1.25,4.5);
	\draw[<-, very thick] (6.25,3) to (6.25,4.5);
	
	\draw[very thick, red] (4.98,0) -- (4.98,1.27);
	\draw[very thick, blue] (5.02,0) -- (5.02,1.23);
	\draw[very thick, red] (4.98,1.27) -- (6.25,1.27);
	\draw[very thick, blue] (5.02,1.23) -- (7.5,1.23);
	\draw[very thick, red] (6.25,1.27) -- (6.25,2.5);
	\draw[very thick, red] (6.25,2.5) -- (7.5,2.5);
	\draw[very thick, blue] (7.5,1.23) -- (7.5,2.5);
	
	\node[rotate=-45] (con) at (6.865,1.865) {$\leftrightsquigarrow_{ch}$};
\end{tikzpicture}
\end{center}
\vspace{-3mm}
$A$, $B$ and $C$ are path shapes, and we would like to compute their pushout. The expected outcome is $D$, since we must identify the three squares by the previous construction. The problem is that the previous construction does not identify $\beta_1$ and $\beta_2$. Those two must be identified because they are both the top right corner of the same square (after identification). We hence need to quotient a little more to be able to define the face maps, as follows. Define $Z_n$ to be the quotient of $Y_n$ by the smallest equivalence relation $\approx$ such that if there are two sequences $u_0, \ldots, u_l$ and $v_0, \ldots, v_l$ such that:\\
\indent --  $[u_0,k] \approx [v_0,k]$,\\
\indent --  for every $0 \leq s \leq l$, $\alpha_{k+1+s}^{u_s} = \alpha_{k+1+s}^{v_s} = 1$,\\
\indent --  for every $0 \leq s < l$, $[u_s,k+s+1] \approx [u_{s+1},k+s+1]$ and $[v_s,k+s+1] \approx [v_{s+1},k+s+1]$,\\
\indent --  $(j_{k+1}^{u_0};1)\star\ldots\star(j_{k+l+1}^{u_l};1) = (j_{k+1}^{v_0};1)\star\ldots\star(j_{k+l+1}^{v_l};1)$,\\
then, $[u_l,k+l+1] \approx [v_l,k+l+1]$. $\col{D}$ is the pHDA $Z_N$ with the face maps being the smallest relations for inclusion such that:\\
\indent --  if $\alpha_k^u = 0$, then $\partial_{j_k^u}^0(\eqcl{u,k})$ is defined and is equal to $\eqcl{u,k-1}$,\\
\indent --  if $\alpha_{k+1}^u = 1$ then $\partial_{j_k^u}^1(\eqcl{u,k})$ is defined and is equal to $\eqcl{u,k+1}$,\\
\indent --  $\partial_{j_1 < \ldots < j_m}^{\beta_1, \ldots, \beta_m}\circ\partial_{i_1 < \ldots < i_n}^{\alpha_1, \ldots, \alpha_n} \subseteq \partial_{k_1 < \ldots < k_{n+m}}^{\gamma_1, \ldots, \gamma_{n+m}}$, for $(k_1, \ldots, k_{n+m} ; \gamma_1, \ldots, \gamma_{n+m}) = (i_1,\ldots, i_n ; \alpha_1, \ldots, \alpha_n) \star (j_1, \ldots, j_m ; \beta_1, \ldots, \beta_m)$.\\
The initial state is $\eqcl{\epsilon}$ and the labelling $\map{\lambda}{\col{D}}{L}$ maps $\eqcl{u,k}$ to $w_k^u$.

\begin{proposition}
\label{prop:colimit}
$\col{D}$ is the colimit of $D$ in $\phda$
\end{proposition}

By \textbf{tree} we mean any pHDA that is the colimit of a diagram with values in path shapes. We denote by $\tr$ the full subcategory of trees.

	\subsection{The unique path properties of trees}
	\label{sub:unique}
	
\noindent\textbf{Failure of the unicity of paths.} Let us consider the pushout square above again. In particular, the pHDA on the bottom-right corner is a tree, by definition. However, there are two paths from $\alpha$ to $\beta$ (in red and blue). This actually comes from the fact that we needed to identify $\beta_1$ and $\beta_2$ to be able to define the face maps. This means that trees do not have the unique path property.\\

\noindent\textbf{Confluent homotopy.} A careful reader may have observed that the only difference between the two previous paths is that some future faces are swapped. Actually, this is the only obstacle for the unicity of paths for trees: there is a unique path modulo equivalence of paths that permutes arrows of the form $\xrightarrow{\_,1}$. That is what we call \textbf{confluent homotopy}. This confluent homotopy will be defined by restricting the elementary homotopies of \cite{vanglabbeek05} to be of only one type out of the four possible, which means our notion of homotopy makes fewer paths equivalent than the one from \cite{vanglabbeek05}. 

\begin{wrapfigure}{r}{0.25\textwidth}
		\vspace{-10mm}
\begin{tikzpicture}[scale=0.8]
	\node (11p) at (5,0) {};
	\node (12p) at (5,3) {};
	\node (21p) at (8,0) {};
	\node (22p) at (8,3) {};
	\draw[draw = white, fill = gray!50] (5,0) rectangle (8,3);
	\path[->,very thick,font=\scriptsize]
		(11p) edge (12p)
		(21p) edge (22p)
		(11p) edge (21p)
		(12p) edge (22p);
	\path[very thick, red]
		(6.5,1.5) edge (6.5,3)
		(6.5,3) edge (8,3);
	\path[very thick, blue]
		(6.5,1.5) edge (8,1.5)
		(8,1.5) edge (8,3);
	\path[very thick, dotted]
		(6,1) edge (6.5,1.5)
		(8,3) edge (8.5,3.5);
	\node at (5,0) {$\bullet$};
	\node at (5,3) {$\bullet$};
	\node at (8,0) {$\bullet$};
	\node at (8,3) {$\bullet$};
	\node at (8.3,2.25) {\textcolor{blue}{$\pi'$}};	
	\node at (7.25,3.3) {\textcolor{red}{$\pi$}};	
	\node[rotate=-45] (con) at (7.25,2.25) {$\leftrightsquigarrow_{ch}$};
\end{tikzpicture}
\vspace{-20mm}
\end{wrapfigure}
We say that a path 
$\pi =  x_0 \xrightarrow{j_1,\alpha_1} x_1 \xrightarrow{j_2,\alpha_2} \ldots \xrightarrow{j_n,\alpha_n} x_n$
is \textbf{elementary confluently homotopic} to a path
$\pi' =  x'_0 \xrightarrow{j'_1,\alpha'_1} x'_1 \xrightarrow{j'_2,\alpha'_2} \ldots \xrightarrow{j'_n,\alpha'_n} x'_n$, and denote by $\pi \leftrightsquigarrow_{ch} \pi'$,
if and only if there are $0 < s < t \leq n$ such that:\\
\noindent -- for all $k < s$ or $k \geq t$, $x_k = x'_k$,\\
\noindent -- for all $k < s$ or $k > t$, $j_k = j'_k$ and $\alpha_k = \alpha'_k$,\\
\noindent -- for all $s \leq k \leq t$, $\alpha_k = \alpha'_k = 1$,\\
\noindent -- $(j_{s},\alpha_{s})\star\ldots\star(j_t,\alpha_t) = (j'_{s},\alpha'_{s})\star\ldots\star(j'_t,\alpha'_t)$.\\
We denote by $\sim_{ch}$, and call \textbf{confluent homotopy}, the reflexive transitive closure of $\leftrightsquigarrow_{ch}$.

\begin{lemma}
\label{lem:colunique}
If $X$ is a tree, then for every element (of any dimension) $x$ of $X$, there is exactly one path modulo confluent homotopy from the initial state to $x$.
\end{lemma}

\vspace{2mm}
	
\noindent\textbf{Shortcuts.} The face maps of path shapes and of the colimits we computed in Section \ref{sub:colimits} are of a very particular form: we start by defining the $\partial_j^\alpha$ and we extend this definition to general $\partial_{j_1 < \ldots < j_n}^{\alpha_1, \ldots, \alpha_n}$. In a way, they are locally defined, and then extended to higher face maps. This means in particular that, in addition to having unique paths modulo confluent homotopy, they also do not have any `shortcut'. A possible shortcut can be defined as a generalisation of paths, in which we allow to make transitions that go, for example, from a point to a square or to a cube, not only to segments, a shortcut being such a possible shortcut which is not confluently homotopic to a path. Those shortcuts may occur in a pHDA, even if it has the unique path property. Concretely, by \textbf{shortcut} we mean the following situation: the face $\partial_{i_1 < \ldots < i_n}^{\alpha_1, \ldots, \alpha_n}(x)$ is defined, but there is no sequence $(j_1;\beta_1)\star\ldots\star(j_n;\beta_n) = (i_1 < \ldots < i_n;\alpha_1, \ldots, \alpha_n)$ such that $\partial_{j_n}^{\alpha_n}\circ\ldots\circ\partial_{j_1}^{\alpha_1}(x)$ is defined. By local-definedness of the face maps:

\begin{lemma}
\label{lem:colloc}
Trees do not have any shortcuts.
\end{lemma}
	
\vspace{2mm}
	
\noindent\textbf{Trees.} We say that a pHDA \textbf{has the unique path property modulo confluent homotopy} if it has no shortcut, and there is exactly one class of paths modulo confluent homotopy from the initial state to any state. Given such a pHDA $X$ and an element $x$ of $X$, by \textbf{depth of $x$} we mean the length of a path from the initial state to $x$ in $X$. Since homotopic paths have the same length, this is uniquely defined. We deduce from the previous discussions that:

\begin{proposition}
\label{prop:coltree}
Trees have unique path property modulo confluent homotopy.
\end{proposition}

In the following, we will prove the converse: trees, defined as colimits of path shapes are exactly those pHDA that have the unique path property modulo confluent homotopy. This will be done by proving that such a pHDA $X$ is isomorphic to its unfolding. A question that occurs now is the following. Much as the general framework of \cite{dubut16}, trees are colimits of paths. Everything tends to work well when those trees have a nice property, which we called \textbf{accessibility}, intuitively, that the colimit process do not `create' paths. This property is actually deeply related to the unicity of paths. Since this unicity fails in the case of pHDA, accessibility fails too. However, an accessibility modulo confluent homotopy holds: the colimit process in pHDA do not create confluent homotopy classes of paths.

	\subsection{Trees are unfoldings}
	\label{sub:unfolding}
	
	We are now constructing our unfolding $U(X)$ of a pHDA $X$ by giving an explicit definition, similar to \cite{vanglabbeek91,fahrenberg13}, and proving that this is a tree. We will prove that there is a covering $\map{\unf_X}{U(X)}{X}$, which in particular means that the unfolding $U(X)$ is $\ps$-bisimilar (in the general sense of \cite{joyal96}) to $X$, and that this covering is actually an isomorphism when $X$ has the unique path property modulo confluent homotopy.\\

\noindent\textbf{Unfolding of a pHDA.} Let us start with a few notations. Given a path $\pi = x_0 \xrightarrow{j_1,\alpha_1} x_1 \xrightarrow{j_2,\alpha_2} \ldots \xrightarrow{j_n,\alpha_n} x_n$ we note $e(\pi) = x_n$, $l(\pi) = n$ and $\pi_{-k} = x_0 \xrightarrow{j_1,\alpha_1} x_1 \xrightarrow{j_2,\alpha_2} \ldots \xrightarrow{j_{n-k},\alpha_{n-k}} x_{n-k}$. Given a pHDA $X$, its \textbf{unfolding} is the following pHDA:\\
\indent -- $U(X)_n$ is the set of equivalence classes $[\pi]$ of paths modulo confluent homotopy, such that $e(\pi)$ is of dimension $n$,\\
\indent -- the face maps are the smallest relations for inclusion such that:\\
\indent\indent $\bullet$  $\partial_i^1(\alpha) = [\pi \xrightarrow{i,1} \partial_i^1(e(\pi))]$, for any $\pi \in \alpha$ such that $\partial_i^1(e(\pi))$ is defined,\\
\indent\indent $\bullet$ $\partial_i^0(\alpha) = [\pi_{-1}]$ for any $\pi \in \alpha$ such that $\pi = \pi_{-1} \xrightarrow{i,0} e(\pi)$,\\
\indent\indent $\bullet$ $\partial_{j_1 < \ldots < j_m}^{\beta_1, \ldots, \beta_m}\circ\partial_{i_1 < \ldots < i_n}^{\alpha_1, \ldots, \alpha_n} \subseteq \partial_{k_1 < \ldots < k_{n+m}}^{\gamma_1, \ldots, \gamma_{n+m}}$, for $(k_1, \ldots, k_{n+m} ; \gamma_1, \ldots, \gamma_{n+m}) = (i_1,\ldots, i_n ; \alpha_1, \ldots, \alpha_n) \star (j_1, \ldots, j_m ; \beta_1, \ldots, \beta_m)$.\\
\indent -- the initial state is $[i]$,\\
\indent -- the labelling is given by $\lambda(\alpha) = \lambda(e(\pi))$ for $\pi \in \alpha$.\\
Following ideas from \cite{dubut16} again, the unfolding can be seen as the glueing of all possible executions of a system, but with care needed to handle confluent homotopy. Concretely:

\begin{proposition}
\label{prop:unftree}
The unfolding of a pHDA is a tree.
\end{proposition}
	
We can also define $\map{\unf_X}{U(X)}{X}$ as the function that maps $[\pi]$ to $e(\pi)$.

\begin{proposition}
\label{prop:unf}
$\unf_X$ is a covering, and so, $U(X)$ is $\ps$-bisimilar to $X$.
\end{proposition}

\noindent\textbf{The unique path property characterises trees.} When $X$ has exactly one class of paths modulo confluent homotopy from the initial state to any state, it is possible to define a function $\map{\eta_X}{X}{U(X)}$ that maps any element $x$ of $X$ to the unique confluent homotopy class to $x$. When furthermore $X$ does not have shortcuts, then $\eta$ is actually a morphism of pHDA. 

\begin{proposition}
\label{prop:natiso}
When $X$ has the unique path property modulo confluent homotopy, then $\eta_X$ is the inverse of $\unf_X$. In particular, $X$ is a tree.
\end{proposition}

Together with Proposition \ref{prop:coltree}, this implies the following:

\begin{theorem}
Trees are exactly the pHDA that have the unique path property modulo confluent homotopy.
\end{theorem}

Another consequence is that this isomorphism $\eta_X$ is actually natural (in the categorical sense) and is part of an adjunction, which implies that trees form a coreflective subcategory of pHDA:



\begin{corollary}
\label{theo:coref}
$U$ extends to a functor, which is the right adjoint of the embedding $\map{\iota}{\tr}{\phda}$. Furthermore, this is a coreflection.
\end{corollary}

\section{Cofibrant objects}
\label{sec:cofi}

Cofibrant objects are another type of `simple objects', coming from homotopy theory, more particularly the language of model categories from \cite{quillen67}. Those cofibrant objects are those whose unique morphism from the initial object is a cofibration. Intuitively (intuition which holds at least in cofibrantly generated model structures \cite{hirschhorn03}), this means that cofibrant objects are those objects constructed from `nothing', using only very basic constructions (generators of cofibrations). In the case of the classical model structure on topological spaces (Kan-Quillen), those spaces are those constructed from the empty space by adding `cells', which produces what is called CW-complexes. In this section, we want to mimic this idea with trees: trees are those pHDA constructed from an initial state by only extending paths. We also want to emphasize that much as CW-complexes gives a kind of homotopy type of a space, trees gives a concurrency type of a pHDA, in the sense that there is a canonical way to produce an equivalent cofibrant object out of any object, which is called the \textbf{cofibrant replacement} in homotopy theory. In concurrency theory, this is the unfolding.

\subsection{Cofibrant objects in $\phda$}

\begin{wrapfigure}{r}{0.21\textwidth}
		\vspace{-9mm}
	\begin{tikzpicture}[scale=1.5]
		\node (A) at (0,0.9) {$\ast$};
		\node (B) at (1.2,0.9) {$Y$};
		\node (C) at (0,0) {$X$};
		\node (D) at (1.2,0) {$Z$};
		\path[->,font=\scriptsize]
		(A) edge node[above]{$!$} (B)
		(B) edge node[right]{$f$} (D)
		(C) edge node[below]{$g$} (D)
		(A) edge node[left]{$!$} (C);
		\path[->,font=\scriptsize,dotted]
		(C) edge node[above]{$h$} (B);
	\end{tikzpicture}
\vspace{-15mm}
\end{wrapfigure}
Following the language of model structures from \cite{quillen67}, we say that a pHDA $X$ is \textbf{cofibrant} if for every $\ps$-open morphism $\map{f}{Y}{Z}$ and every morphism $\map{g}{X}{Z}$, there is a morphism $\map{h}{X}{Y}$, such that $f\circ h = g$. That is, a partial HDA $X$ is cofibrant if and only if every $\ps$-open morphism has the right lifting property with respect to the unique morphism from $\ast$ to $X$. 
	
	\subsection{Cofibrant objects are exactly trees}
	
	In this section, we would like to prove the following:
	
	\begin{theorem}
	The cofibrant objects are exactly trees.
	\end{theorem}
	
	\begin{wrapfigure}{r}{0.21\textwidth}
		\vspace{-9mm}
	\begin{tikzpicture}[scale=1.5]
		\node (A) at (0,0.9) {$\ast$};
		\node (B) at (1.2,0.9) {$U(X)$};
		\node (C) at (0,0) {$X$};
		\node (D) at (1.2,0) {$X$};
		\path[->,font=\scriptsize]
		(A) edge node[above]{$!$} (B)
		(B) edge node[right]{unf} (D)
		(C) edge node[below]{$\text{id}_X$} (D)
		(A) edge node[left]{$!$} (C);
		\path[->,font=\scriptsize,dotted]
		(C) edge node[above]{$h$} (B);
	\end{tikzpicture}
\vspace{-15mm}
\end{wrapfigure}
	 Let us start by giving the idea of the proof of the fact that cofibrant objects are trees. By Proposition \ref{prop:unf}, $\unf_X$ is a covering, so is open. This means that for every cofibrant object $X$, there is a morphism $\map{h}{X}{U(X)}$ such that $\unf_X\circ h = \text{id}_X$, that is, $X$ is a retract of its unfolding. Since we know that the unfolding is a tree by Proposition \ref{prop:unftree}, it is enough to observe the following:

	\begin{lemma}
	\label{lem:retract}
	A retract of a tree is a tree.
	\end{lemma}
	
	Intuitively, a pHDA is the retract of a tree only when it is obtain by retracting branches. This can only produce a tree. For the converse:
	
	\begin{proposition}
	\label{prop:cofi}
	A tree is a cofibrant object. Furthermore, if $\map{f}{Y}{Z}$ is a covering, then the lift $\map{h}{X}{Y}$ is unique.
	\end{proposition}

	The lift $h$ is constructed by induction as follows. We define $X_n$ as the restriction of $X$ to elements whose depth is smaller than $n$, and the face maps $\partial_{j_1 < \ldots < j_m}^{\alpha_1, \ldots, \alpha_m}(x)$ are defined if and only if $\partial_{j_1 < \ldots < j_m}^{\alpha_1, \ldots, \alpha_m}(x)$ is defined in $X$ and belongs to $X_n$. We then construct $\map{h_n}{X_n}{Y}$ using the unique path property modulo confluent homotopy, in a natural way (in the categorical meaning), i.e., such that $h_{n}\circ\kappa_n = h_{n-1}$, where $\map{\kappa_n}{X_{n-1}}{X_{n}}$ is the inclusion. $h$ is then the inductive limit of those $h_n$. This proof can be seen as a small object argument.

	\subsection{The unfolding is universal}
	
	As an application of the previous theorem, we would like to prove that the unfolding is universal. As in the case of covering spaces in algebraic topology, a covering corresponds to a partial unrolling of a system, in the sense that we can unroll some loops or even partially unroll a loop (imagine for example executing a few steps of a while-loop). In this sense, we can describe the fact that a covering unrolls more than another one, and that, an unfolding is a complete unrolling: since the domain is a tree, it is impossible to unroll more. Actually, much as the topological and the groupoidal cases (see \cite{may99} for example), unfoldings are the only such maximal unrollings among coverings: they are initial among coverings, that is why we call them `universal'. In a way, this says that our definition of unfolding is the only reasonable one. Concretely, we say that a $\ps$-covering is \textbf{universal} if its domain is a tree.

	\begin{corollary}
	\label{prop:univ}
	If $\map{f}{Y}{X}$ is a universal covering, then for every covering $\map{g}{Z}{X}$ there is a unique map $\map{h}{Y}{X}$ such that $f = g\circ h$. Furthermore, $h$ is itself a covering. Consequently, the universal covering is unique up-to isomorphism, and is given by the unfolding.
	\end{corollary}
	
	This whole story is similar to the universal covering of a topological space: just replace pHDA by spaces and trees by simply-connected spaces \cite{tomdieck08}.

\section{Conclusion and future work}

In this paper, we have given a cleaner definition of partial precubical sets and partial Higher Dimensional Automata, as they really correspond to collections of cubes with missing faces. From this categorical definition, we derived that pHDA can be completed, giving rise to a geometric realisation. We also describe the first premisses of a homotopy theory of the concurrency of pHDA where the cofibrant objects are trees, and replacement is the unfolding.
As a future work, we could look at wider class of paths, typically allowing shortcuts as paths, or introducing general homotopies in the path category, which is possible because we can encode those inside the category of pHDA. Another direction would be to continue the description of this homotopy theory, to see if it corresponds to some kind of Quillen's model structure, or at least to some weaker version (e.g., category of cofibrant objects).

%
%
%
 \bibliographystyle{splncs04}
 \bibliography{partialhda}

\begin{thebibliography}{10}
\providecommand{\url}[1]{\texttt{#1}}
\providecommand{\urlprefix}{URL }
\providecommand{\doi}[1]{https://doi.org/#1}

\bibitem{bednarczyk87}
Bednarczyk, M.A.: Categories of asynchronous systems. Phd thesis, University of
  Sussex (1987)

\bibitem{tomdieck08}
tom Dieck, T.: Algebraic Topology. Textbooks in Mathematics, European
  Mathematical Society (2008)

\bibitem{dubut17}
Dubut, J.: Directed homotopy and homology theories for geometric models of true
  concurrency. Ph.D. thesis, ENS Paris-Saclay (2017)

\bibitem{dubut16}
Dubut, J., Goubault, E., Goubault-Larrecq, J.: Bisimulations and {U}nfolding in
  {P}-{A}ccessible {C}ategorical {M}odels. In: Proceedings of the 27th
  {I}nternational {C}onference on {C}oncurrency {T}heory ({CONCUR} 2016).
  Leibniz {I}nternational {P}roceedings in {I}nformatics ({LIPI}cs), vol.~59,
  pp. 1--14. Schloss {D}agstuhl--{L}eibniz-{Z}entrum fuer {I}nformatik (2016)

\bibitem{fahrenberg05}
Fahrenberg, U.: A category of higher-dimensional automata. In: Proceedings of
  the 8th conference on the Foundations of Software Science and Computational
  Structures ({FOSSACS} 2005). Lecture Notes in Computer Science, vol.~3441,
  pp. 187--201. Springer (2005)

\bibitem{fahrenberg13}
Fahrenberg, U., Legay, A.: {H}istory-{P}reserving {B}isimilarity for
  {H}igher-{D}imensional {A}utomata via {O}pen {M}aps. Electronic Notes in
  Theoretical Computer Science  \textbf{298},  165--178 (November 2013)

\bibitem{fahrenberg15}
Fahrenberg, U., Legay, A.: Partial {Hi}gher-dimensional {A}utomata. In:
  CALCO'15. pp. 101--115 (2015)

\bibitem{fajstrup05}
Fajstrup, L.: Dipaths and dihomotopies in a cubical complex. Adances in
  {A}pplied {M}athematics  \textbf{35}(2),  188--206 (August 2005)

\bibitem{fajstrup16}
Fajstrup, L., Goubault, E., Haucourt, E., Mimram, S., Raussen, M.: Directed
  {A}lgebraic {T}opology and {C}oncurrency. Springer (2016)

\bibitem{gaucher11}
Gaucher, P.: Towards a homotopy theory of higher dimensional transition
  systems. Theory and {A}pplications of {C}ategories  \textbf{25},  295--341
  (2011)

\bibitem{vanglabbeek91}
van Glabbeek, R.J.: Bisimulations for higher dimensional automata (June 1991),
  email message, \url{http://theory.stanford.edu/~rvg/hda}

\bibitem{vanglabbeek05}
van Glabbeek, R.J.: On the {E}xpresiveness of {H}igher {D}imensional
  {A}utomata: ({E}xtended {A}bstract). Electronic Notes in Theoretical Computer
  Science  \textbf{128}(2),  5--34 (April 2005)

\bibitem{goubault95}
Goubault, E.: G\'eom\'etrie du parall\'elisme. Phd thesis, Ecole
  {P}olytechnique (1995)

\bibitem{grandis09}
Grandis, M.: Directed {A}lgebraic {T}opology: {M}odels of {N}on-{R}eversible
  {W}orlds, New Mathematical Monographs, vol.~13. Cambridge University Press
  (2009)

\bibitem{hirschhorn03}
Hirschhorn, P.S.: Model {C}ategories and {T}heir {L}ocalizations, Mathematical
  Surveys and Monographs, vol.~99. American Mathematical Society (2003)

\bibitem{jacobs16}
Jacobs, B.: Introduction to {C}oalgebra: {T}owards {M}athematics of {S}tates
  and {O}bservation. Cambridge {T}racts in {T}heoretical {C}omputer {S}cience,
  Cambridge University Press (October 2016)

\bibitem{joyal96}
Joyal, A., Nielsen, M., Winskel, G.: Bisimulation from open maps. Information
  and {C}omputation  \textbf{127}(2),  164--185 (June 1996)

\bibitem{may99}
May, J.P.: A {C}oncise {C}ourse in {A}lgebraic {T}opology. Chicago Lectures in
  Mathematics, University of Chicago Press (1999)

\bibitem{milner80}
Milner, R.: A {C}alculus of {C}ommunication {S}ystems, Lecture Notes in
  Computer Science, vol.~92. Springer (1980)

\bibitem{niefield10}
Niefield, S.: Lax presheaves and exponentiability. Theory and {A}pplications of
  {C}ategories  \textbf{24}(12),  288--301 (2010)

\bibitem{nielsen81}
Nielsen, M., Plotkin, G., Winskel, G.: Petri nets, event structures and
  domains, part {I}. Theoretical Computer Science  \textbf{13}(1),  85--108
  (1981)

\bibitem{nielsen94}
Nielsen, M., Sassone, V., Winskel, G.: Relationships {B}etween {M}odels of
  {C}oncurrency. In: A Decade of Concurrency Reflections and Perspectives: REX
  School/Symposium Noordwijkerhout, The Netherlands June 1--4, 1993
  Proceedings. pp. 425--476. Springer Berlin Heidelberg (1994)

\bibitem{pratt91}
Pratt, V.: Modeling concurrency with geometry. In: Proceedings of the 18th {ACM
  SIGPLAN-SIGACT} symposium on {P}rinciples of programming languages ({POPL}).
  pp. 311--322 (January 1991)

\bibitem{quillen67}
Quillen, D.G.: Homotopical {A}lgebra, Lecture Notes in Mathematics, vol.~43.
  Springer Berlin Heidelberg (1967)

\bibitem{shields85}
Shields, M.W.: Concurrent {M}achines. Computer {J}ournal  \textbf{28}(5),
  449--465 (January 1985)

\end{thebibliography}

\appendix

\section{Omitted proofs of Section \ref{sec:phda}}

\subsection*{Proof of Theorem \ref{the:reflection}}

\begin{lemma}
\label{lem:completion}
$\chi(X)$ is a well defined HDA.
\end{lemma}

\begin{proof}~
\begin{enumerate}
	\item \textbf{$\chi(X)$ is well defined as a precubical set.} 
		\begin{itemize}
			\item \textbf{The face maps are well defined.} This means that if 
			$$((i_1 < \ldots < i_{n} ; \alpha_1, \ldots, \alpha_n),x) \sim ((j_1 < \ldots < j_l ; \beta_1, \ldots, \beta_l),y)$$ 
			then 
			$$\partial_i^\alpha(\ll(i_1 < \ldots < i_{n} ; \alpha_1, \ldots, \alpha_n),x\gg) = \partial_i^\alpha(\ll(j_1 < \ldots < j_l ; \beta_1, \ldots, \beta_l),y\gg)$$ 
			which is equivalent, by definition of $\partial_i^\alpha$, to 
			$$((i_1 < \ldots < i_{n} ; \alpha_1, \ldots, \alpha_n)\star(i;\alpha),x) \sim ((j_1 < \ldots < j_l ; \beta_1, \ldots, \beta_l)\star(i;\alpha),y).$$ 
			This is given by the second point in the definition of $\sim$.
			\item \textbf{The local equations hold.} We want to prove that 
			$$\partial_{i}^\alpha\circ\partial_{j}^\beta = \partial_{j}^\beta\circ\partial_{i+1}^\alpha$$ 
			for every $j \leq i$. This comes from the fact that 
			$$(j;\beta)\star(i;\alpha) = (j < i+1 ; \beta,\alpha) = (i+1;\alpha)\star(j,\beta)$$ 
			by definition of $\star$.
		\end{itemize}
	\item \textbf{$\chi(X)$ is well defined as a HDA.}
		\begin{itemize}
			\item \textbf{The labelling function is well defined.} We want to prove that if 
			$$((i_1 < \ldots < i_{n} ; \alpha_1, \ldots, \alpha_n),x) \sim ((j_1 < \ldots < j_l ; \beta_1, \ldots, \beta_l),y)$$ 
			then 
			$$\lambda(\ll(i_1 < \ldots < i_{n} ; \alpha_1, \ldots, \alpha_n),x\gg) = \lambda(\ll(j_1 < \ldots < j_l ; \beta_1, \ldots, \beta_l),y\gg).$$ 
			Let us prove it by induction.
				\begin{itemize}
					\item[$\bullet$] \textit{Base case.} We want to prove that if $\partial_{i_1 < \ldots < i_n}^{\alpha_1, \ldots, \alpha_n}(x)$ is well defined in the pHDA $X$, then 
					$$\lambda(\ll(i_1 < \ldots < i_{n} ; \alpha_1, \ldots, \alpha_n),x\gg) = \lambda(\ll\epsilon,\partial_{i_1 < \ldots < i_n}^{\alpha_1, \ldots, \alpha_n}(x)\gg).$$
					This is proved as follows:
\begin{center}					
\begin{tabular}{clcr}
     & $\lambda(\ll(i_1 < \ldots < i_{n} ; \alpha_1, \ldots, \alpha_n),x\gg)$ & ~ & \\
     $=$ & $\delta_{i_1}^{\alpha_1}\circ\ldots\circ\delta_{i_n}^{\alpha_n}(\lambda(x))$ & ~ & definition of $\lambda$ \\
     $=$ & $\delta_{i_1 < \ldots < i_n}^{\alpha_1, \ldots, \alpha_n}(\lambda(x))$ & ~ & definition of $\delta_{i_1 < \ldots < i_n}^{\alpha_1, \ldots, \alpha_n}$\\
     $=$ & $\lambda(\partial_{i_1 < \ldots < i_n}^{\alpha_1, \ldots, \alpha_n}(x))$ & ~ & $\lambda$ morphism of pHDA\\
     $=$ & $\lambda(\ll\epsilon,\partial_{i_1 < \ldots < i_n}^{\alpha_1, \ldots, \alpha_n}(x)\gg)$ & ~ & definition of $\lambda$
\end{tabular}
\end{center}
					\item[$\bullet$] \textit{Induction case.} We want to prove that if 
					$$\lambda(\ll(i_1 < \ldots < i_{n} ; \alpha_1, \ldots, \alpha_n),x\gg) = \lambda(\ll(j_1 < \ldots < j_l ; \beta_1, \ldots, \beta_l),y\gg)$$ 
					then 
					$$\lambda(\ll(i_1 < \ldots < i_{n} ; \alpha_1, \ldots, \alpha_n)\star(i;\alpha),x\gg) = \lambda(\ll(j_1 < \ldots < j_l ; \beta_1, \ldots, \beta_l)\star(i,\alpha),y\gg).$$ 
					This is proved as follows, for 
					$$(i_1 < \ldots < i_{n} ; \alpha_1, \ldots, \alpha_n)\star(i;\alpha) = (i'_1 < \ldots < i'_{n+1} ; \alpha'_1, \ldots, \alpha'_{n+1})$$ 
					and 
					$$(j_1 < \ldots < j_l ; \beta_1, \ldots, \beta_l)\star(i;\alpha) = (j'_1 < \ldots < j'_{l+1} ; \beta'_1, \ldots, \beta'_{l+1})$$
					
					\begin{center}
					\begin{tabular}{clcr}
     & $\lambda(\ll(i_1 < \ldots < i_{n} ; \alpha_1, \ldots, \alpha_n)\star(i;\alpha),x\gg)$ & ~ & \\
     $=$ & $\delta_{i'_1}^{\alpha'_1}\circ\ldots\circ\delta_{i'_{n+1}}^{\alpha'_{n+1}}(\lambda(x))$ & & definition of $\lambda$\\
     $=$ & $\delta_i^\alpha\circ\delta_{i_1}^{\alpha_1}\ldots\circ\delta_{i_n}^{\alpha_n}(\lambda(x))$ & & definition of $\star$ and local equations for $\delta$\\
     $=$ & $\delta_i^\alpha(\lambda(\ll(i_1 < \ldots < i_{n} ; \alpha_1, \ldots, \alpha_n),x\gg))$ & & definition of $\lambda$\\
     $=$ & $\delta_i^\alpha(\lambda(\ll(j_1 < \ldots < j_l ; \beta_1, \ldots, \beta_l),y\gg))$ & & induction hypothesis\\
     $=$ & $\delta_i^\alpha\circ\delta_{j_1}^{\beta_1}\ldots\circ\delta_{j_l}^{\beta_l}(\lambda(x))$ & & definition of $\lambda$\\
     $=$ & $\delta_{j'_1}^{\beta'_1}\circ\ldots\circ\delta_{j'_{l+1}}^{\beta'_{l+1}}(\lambda(x))$ & & definition of $\star$ and local equations for $\delta$\\
     $=$ & $\lambda(\ll(j_1 < \ldots < j_l ; \beta_1, \ldots, \beta_l)\star(i,\alpha),y\gg)$ & & by definition of $\lambda$
\end{tabular}
\end{center}
					
				\end{itemize}
			\item \textbf{The labelling equation holds.} We want to prove that, for $i \in \{1,2\}$: 
			$$\lambda\circ\partial_i^0 = \lambda\circ\partial_i^1.$$ 
			This is proved as follows:
			
			\begin{center}
			\begin{tabular}{clcr}
     & $\lambda(\partial_i^0(\ll(i_1 < \ldots < i_{n} ; \alpha_1, \ldots, \alpha_n),x\gg))$ & ~ & \\
     $=$ & $\lambda(\ll(i_1 < \ldots < i_{n} ; \alpha_1, \ldots, \alpha_n)\star(i;0),x\gg)$ & & definition of $\partial_i^0$\\
     $=$ & $\delta_i^0\circ\delta_{i_1}^{\alpha_1}\ldots\circ\delta_{i_n}^{\alpha_n}(\lambda(x))$ & & definition of $\lambda$\\
     $=$ & $\delta_i^1\circ\delta_{i_1}^{\alpha_1}\ldots\circ\delta_{i_n}^{\alpha_n}(\lambda(x))$ & & $\delta_i^0 = \delta_i^1$\\
     $=$ & $\lambda(\ll(i_1 < \ldots < i_{n} ; \alpha_1, \ldots, \alpha_n)\star(i;1),x\gg)$ & & definition of $\lambda$\\
     $=$ & $\lambda(\partial_i^1(\ll(i_1 < \ldots < i_{n} ; \alpha_1, \ldots, \alpha_n),x\gg))$ & & definition of $\partial_i^1$
			\end{tabular}
			\end{center}
		\end{itemize}
\end{enumerate}
\end{proof}

\begin{lemma} Define $$\chi(f)(\ll(i_1 < \ldots < i_k;\alpha_1, \ldots, \alpha_k),x\gg) = \ll(i_1 < \ldots < i_k;\alpha_1, \ldots, \alpha_k),f(x)\gg.$$
$\chi(f)$ is a well defined morphism of HDA.
\end{lemma}

\begin{proof}~
\begin{enumerate}
	\item \textbf{$\chi(f)$ is well defined.} We want to prove that if 
	$$((i_1 < \ldots < i_{n} ; \alpha_1, \ldots, \alpha_n),x) \sim ((j_1 < \ldots < j_l ; \beta_1, \ldots, \beta_l),y)$$ 
			then 
			$$\chi(f)(\ll(i_1 < \ldots < i_{n} ; \alpha_1, \ldots, \alpha_n),x\gg) = \chi(f)(\ll(j_1 < \ldots < j_l ; \beta_1, \ldots, \beta_l),y\gg).$$ 
			Let us prove it by induction.
			\begin{itemize}
					\item[$\bullet$] \textit{Base case.} We want to prove that if $\partial_{i_1 < \ldots < i_n}^{\alpha_1, \ldots, \alpha_n}(x)$ is well defined in the pHDA $X$, then 
					$$\chi(f)(\ll(i_1 < \ldots < i_{n} ; \alpha_1, \ldots, \alpha_n),x\gg) = \chi(f)(\ll\epsilon,\partial_{i_1 < \ldots < i_n}^{\alpha_1, \ldots, \alpha_n}(x)\gg).$$
					This is proved as follows. First, observe that since $f$ is a morphism of pHDA, $\partial_{i_1 < \ldots < i_n}^{\alpha_1, \ldots, \alpha_n}(f(x))$ is well defined in the pHDA $Y$ and is equal to $f(\partial_{i_1 < \ldots < i_n}^{\alpha_1, \ldots, \alpha_n}(x))$.

\begin{center}				
\begin{tabular}{clcr}
     & $\chi(\ll(i_1 < \ldots < i_{n} ; \alpha_1, \ldots, \alpha_n),x\gg)$ & ~ & \\
     $=$ & $\ll(i_1 < \ldots < i_{n} ; \alpha_1, \ldots, \alpha_n),f(x)\gg$ & & definition of $\chi$ \\
     $=$ & $\ll\epsilon,\partial_{i_1 < \ldots < i_n}^{\alpha_1, \ldots, \alpha_n}(f(x))\gg$ & & definition of $\sim$ and $f$ morphism of pHDA\\
     $=$ & $\ll\epsilon,f(\partial_{i_1 < \ldots < i_n}^{\alpha_1, \ldots, \alpha_n}(x))\gg$ & & $f$ morphism of pHDA\\
     $=$ & $\chi(f)(\ll\epsilon,\partial_{i_1 < \ldots < i_n}^{\alpha_1, \ldots, \alpha_n}(x)\gg)$ & & definition of $\chi$
\end{tabular}
\end{center}
					\item[$\bullet$] \textit{Induction case.} We want to prove that if 
					$$\chi(f)(\ll(i_1 < \ldots < i_{n} ; \alpha_1, \ldots, \alpha_n),x\gg) = \chi(f)(\ll(j_1 < \ldots < j_l ; \beta_1, \ldots, \beta_l),y\gg)$$ 
					then 
					$$\chi(f)(\ll(i_1 < \ldots < i_{n} ; \alpha_1, \ldots, \alpha_n)\star(i;\alpha),x\gg) = \chi(f)(\ll(j_1 < \ldots < j_l ; \beta_1, \ldots, \beta_l)\star(i,\alpha),y\gg).$$ 
					
\begin{center}
\begin{tabular}{clrc}
     & $\chi(f)(\ll(i_1 < \ldots < i_{n} ; \alpha_1, \ldots, \alpha_n)\star(i;\alpha),x\gg)$ & ~ & \\
     $=$ & $\ll(i_1 < \ldots < i_{n} ; \alpha_1, \ldots, \alpha_n)\star(i;\alpha),f(x)\gg$ & & definition of $\chi$\\
     $=$ & $\partial_i^\alpha(\ll(i_1 < \ldots < i_{n} ; \alpha_1, \ldots, \alpha_n),f(x)\gg)$ & & definition of $\partial$\\
     $=$ & $\partial_i^\alpha(\chi(f)(\ll(i_1 < \ldots < i_{n} ; \alpha_1, \ldots, \alpha_n),x\gg))$ & & definition of $\chi$\\
     $=$ & $\partial_i^\alpha(\chi(f)(\ll(j_1 < \ldots < j_l ; \beta_1, \ldots, \beta_l),y\gg))$ & & induction hypothesis\\
     $=$ & $\chi(f)(\ll(j_1 < \ldots < j_l ; \beta_1, \ldots, \beta_l)\star(i,\alpha),y\gg)$ & & similarly
\end{tabular}
\end{center}
					
				\end{itemize}
	\item \textbf{$\chi(f)$ is a morphism of precubical sets.} We want to prove that 
	$$\chi(f)\circ\partial_i^\alpha = \partial_i^\alpha\circ\chi(f).$$
	
					\begin{tabular}{rcl}
    $\chi(f)(\partial_i^\alpha(\ll(i_1 < \ldots < i_{n} ; \alpha_1, \ldots, \alpha_n),x\gg))$ & $=$ & $\ll(i_1 < \ldots < i_{n} ; \alpha_1, \ldots, \alpha_n)\star(i;\alpha),f(x)\gg$\\
     & $=$ & $\partial_i^\alpha(\chi(f)(\ll(i_1 < \ldots < i_{n} ; \alpha_1, \ldots, \alpha_n),x\gg))$
\end{tabular}
	
	\item \textbf{$\chi(f)$ is a morphism of HDA.} We want to prove that 
	$$\lambda\circ\chi(f) = \lambda.$$

\begin{center}
\begin{tabular}{clcr}
     & $\lambda(\chi(f)(\ll(i_1 < \ldots < i_{n} ; \alpha_1, \ldots, \alpha_n),x\gg))$ & ~ & \\
     $=$ & $\lambda(\ll(i_1 < \ldots < i_{n} ; \alpha_1, \ldots, \alpha_n),f(x)\gg)$ & & definition of $\chi$\\
     $=$ & $\delta_{i_1 < \ldots < i_n}^{\alpha_1, \ldots, \alpha_n}(\lambda(f(x)))$ & & definition of $\lambda$\\
     $=$ & $\delta_{i_1 < \ldots < i_n}^{\alpha_1, \ldots, \alpha_n}(\lambda(x))$ & & $f$ morphism of pHDA\\
     $=$ & $\lambda(\ll(i_1 < \ldots < i_{n} ; \alpha_1, \ldots, \alpha_n),x\gg)$ & & definition of $\lambda$
\end{tabular}
\end{center}
	
\end{enumerate}
\end{proof}

\begin{lemma}
$\map{\mu_X}{\chi X}{X}$, which maps $\ll(i_1 < \ldots < i_k; \alpha_1, \ldots, \alpha_k), x \gg$ to $\partial_{i_k}^{\alpha_k}\circ\ldots\circ\partial_{i_1}^{\alpha_1}(x)$ forms a natural isomorphism $\natr{\mu}{\chi\circ\tau}{\text{id}_{\hda}}$, whose inverse is $\map{\theta_X}{X}{\chi X}$, which maps $x$ to $\ll\epsilon, x\gg$.
\end{lemma}

\begin{proof}~
\begin{enumerate}
	\item \textbf{$\mu_X$ is well defined for every HDA $X$.} We want to prove that if 
	$$((i_1 < \ldots < i_{n} ; \alpha_1, \ldots, \alpha_n),x) \sim ((j_1 < \ldots < j_l ; \beta_1, \ldots, \beta_l),y)$$ 
			then 
			$$\mu_X(\ll(i_1 < \ldots < i_{n} ; \alpha_1, \ldots, \alpha_n),x\gg) = \epsilon_X(\ll(j_1 < \ldots < j_l ; \beta_1, \ldots, \beta_l),y\gg).$$ 
			Let us prove it by induction.
			\begin{itemize}
					\item[$\bullet$] \textit{Base case.} We want to prove that if $\partial_{i_1 < \ldots < i_n}^{\alpha_1, \ldots, \alpha_n}(x)$ is well defined in the HDA $X$ (which is always the case in a HDA), then 
					$$\mu_X(\ll(i_1 < \ldots < i_{n} ; \alpha_1, \ldots, \alpha_n),x\gg) = \mu_X(\ll\epsilon,\partial_{i_1 < \ldots < i_n}^{\alpha_1, \ldots, \alpha_n}(x)\gg).$$
			Both are equal to $\partial_{i_1 < \ldots < i_n}^{\alpha_1, \ldots, \alpha_n}(x)$.
					\item[$\bullet$] \textit{Induction case.} We want to prove that if 
					$$\mu_X(\ll(i_1 < \ldots < i_{n} ; \alpha_1, \ldots, \alpha_n),x\gg) = \mu_X(\ll(j_1 < \ldots < j_l ; \beta_1, \ldots, \beta_l),y\gg)$$ 
					then 
					$$\mu_X(\ll(i_1 < \ldots < i_{n} ; \alpha_1, \ldots, \alpha_n)\star(i;\alpha),x\gg) = \epsilon_X(\ll(j_1 < \ldots < j_l ; \beta_1, \ldots, \beta_l)\star(i,\alpha),y\gg).$$ 
					This is proved as follows:
					
\begin{center}
\begin{tabular}{clcr}
	& $\mu_X(\ll(i_1 < \ldots < i_{n} ; \alpha_1, \ldots, \alpha_n)\star(i;\alpha),x\gg)$ & ~ & \\
	$=$ & $\partial_i^\alpha\circ\partial_{i_1 < \ldots < i_n}^{\alpha_1, \ldots, \alpha_n}(x)$ & & definition of $\mu_X$\\
    	$=$ & $\partial_i^\alpha\circ\partial_{j_1 < \ldots < j_l}^{\beta_1, \ldots, \beta_l}(y)$ & & induction hypothesis\\
	$=$ & $\mu_X(\ll(j_1 < \ldots < j_l ; \beta_1, \ldots, \beta_l)\star(i;\alpha),y\gg)$ & & definition of $\mu_X$
\end{tabular}
\end{center}
					
				\end{itemize}
	\item \textbf{$\mu_X$ is a morphism of precubical sets.} We want to prove that 
	$$\mu_X\circ\partial_i^\alpha = \partial_i^\alpha\circ\mu_X$$
	whose proof is the same as the induction case of the previous point.
	\item \textbf{$\mu_X$ is a morphism of HDA.} We want to prove that
	$$\lambda\circ\mu_X = \lambda$$
	which is equivalent to the fact that the labelling function of $X$ is a morphism of precubical sets.
	\item \textbf{$\theta_X$ is a morphism of precubical sets.} This comes from the fact that $$(\epsilon,\partial_i^\alpha(x)) \sim ((i;\alpha),x).$$
	\item \textbf{$\theta_X$ is a morphism of HDA.} This is by definition of the labelling function of $\chi X$.
	\item \textbf{$\mu_X$ and $\theta_X$ are inverse of each other.}
		\begin{itemize} 
			\item $\theta_X\circ\mu_X = \text{id}_{\chi X}$. This is a consequence of 
			$$((i_1 < \ldots < i_n;\alpha_1, \ldots, \alpha_n),x) \sim (\epsilon, \partial_{i_1 < \ldots < i_n}^{\alpha_1, \ldots, \alpha_n}(x)).$$
			\item $\mu_X\circ\theta_X = \text{id}_X$. Obvious.
		\end{itemize}
	\item \textbf{$\mu$ is natural.} Let $\map{f}{X}{Y}$ be a morphism of HDA. We want to prove that
	$$f \circ \mu_X = \mu_Y\circ\chi(f).$$

\begin{center}
\begin{tabular}{clcr}
	& $f(\mu_X(\ll(i_1 < \ldots < i_{n} ; \alpha_1, \ldots, \alpha_n),x\gg))$ & ~ & \\
	$=$ & $f\circ\partial_{i_1 < \ldots < i_n}^{\alpha_1, \ldots, \alpha_n}(x)$ & & definition of $\mu_X$\\
	$=$ & $\partial_{i_1 < \ldots < i_n}^{\alpha_1, \ldots, \alpha_n}(f(x))$ & & $f$ morphism of precubical sets\\
	$=$ & $\mu_X(\ll(i_1 < \ldots < i_{n} ; \alpha_1, \ldots, \alpha_n),f(x)\gg)$ & & definition of $\mu_X$\\
	$=$ & $\mu_X(\chi(f)(\ll(i_1 < \ldots < i_{n} ; \alpha_1, \ldots, \alpha_n),x\gg))$ & & definition of $\chi(f)$
\end{tabular}
\end{center}
\end{enumerate}
\end{proof}

\begin{lemma}
$\natr{\eta_X}{X}{\chi X}$, which maps $x$ to $\ll\epsilon, x\gg$, forms a natural transformation $\natr{\eta}{\text{id}_{\phda}}{\tau\circ\chi}$. Observe that the definition of $\eta$ is the same as the definition of $\theta$, except that $\eta$ is defined for every pHDA, while $\theta$ is defined only on HDA.
\end{lemma}

\begin{proof}~
\begin{enumerate}
	\item \textbf{$\eta_X$ is a morphism of HDA.} Similar to $\theta$.
	\item \textbf{$\eta$ is natural.} $\chi(f)\circ\eta_X(x) = \ll\epsilon, f(x)\gg = \eta_Y\circ f(x)$.
\end{enumerate}
\end{proof}

\begin{proof}[of Theorem \ref{the:reflection}] We prove that $\epsilon$ is the counit and that $\eta$ is the unit.
\begin{enumerate}
	\item Let $X$ be a pHDA.
	
\begin{center}
\begin{tabular}{clcr}
	& $\epsilon_{\chi X}\circ\chi(\eta_X)(\ll(i_1 < \ldots< i_k;\alpha_1, \ldots, \alpha_k),x\gg)$ & ~ & \\
	$=$ & $\epsilon_{\chi X}(\ll (i_1 < \ldots < i_k;\alpha_1, \ldots, \alpha_k),\eta_X(x)\gg)$ & & definition of $\chi$\\
	$=$ & $\partial_{i_1 < \ldots < i_k}^{\alpha_1, \ldots, \alpha_k}(\ll \epsilon, x\gg)$ & & by definition of $\mu$ and $\eta$ \\
	$=$ & $\ll(i_1 < \ldots < i_k;\alpha_1, \ldots, \alpha_k), x\gg$ & & by definition of $\partial$ and $\star$
\end{tabular}
\end{center}
	\item Let $X$ be a HDA. $\mu_X\circ\eta_X = \mu_X\circ\theta_X = \text{id}_X$.
\end{enumerate}
\end{proof}

\section{Omitted proofs of Section \ref{sec:path}}

\subsection*{Proof of Proposition \ref{prop:morph}}

If $\alpha_k = 0$, then $\pi$ being a path means that $x_{k-1} = \partial_{j_k}^0(x_k)$. So $$f(x_{k-1}) = f(\partial_{j_k}^0(x_k)) = \partial_{j_k}^0(f(x_k))$$
which is what is needed to prove that $\pi'$ is a path. Remark that the second equality holds for our definition of morphisms, not for the one from \cite{fahrenberg15}, since $\partial_{j_k}^0(f(x_k))$ would not be required to be defined.
The same holds if $\alpha_k = 1$.

\subsection*{Proof of well definedness of $B\sigma$}

Observe that the face maps are defined as the smallest relation such that some conditions hold. We have to check first that this relation is actually functional. It is enough to prove that if $(i_1;\alpha_1)\star\ldots\star(i_n;\alpha_n) = (j_1;\beta_1)\star\ldots\star(j_n,\beta_n)$ and if both $\partial_{i_n}^{\alpha_n}\circ\ldots\circ\partial_{i_1}^{\alpha_1}(k)$ and $\partial_{j_n}^{\beta_n}\circ\ldots\circ\partial_{j_1}^{\beta_1}(k)$, then they are necessarily equal.

Observe that $\partial_j^0\partial_i^1(k)$ and $\partial_j^1\partial_i^0(k)$ are never defined. If, for example, the first were, that would mean that $\partial^1(k)$ is defined to be $k+1$ and we have the transition $(d_k,w_k) \xrightarrow{i,1} (d_{k+1},w_{k+1})$ in $\sigma$. But we also have that $\partial_j^0(k+1)$ is defined, which means that we have the transition $(d_k,w_k) \xrightarrow{j,0} (d_{k+1},w_{k+1})$ in $\sigma$, which is impossible. This means that the $\alpha_i$ and all the $\beta_j$ are either all equal to $0$, or all equal to $1$. Let us assume they are equal to $1$. Then, the sequences $i_1$, \ldots, $i_n$ and $j_1$, \ldots, $j_n$ are both equal to the the sequence $l_1$, \ldots, $l_n$ such that there are transitions $(d_k,w_k) \xrightarrow{l_1,1} (d_{k+1},w_{k+1}) \ldots \xrightarrow{l_n,1} (d_{k+n},w_{k+n})$

The fact that it is a partial precubical set is the third condition of the definition of the face maps. The fact that $\lambda$ is a map of partial precubical sets is by definition of a spine.

\subsection*{Proof of Proposition \ref{prop:bij}}

A bijection is given by the following:
\begin{itemize}
	\item from a path $\pi$ of $X$, construct a morphism from $B\sigma_\pi$ to $X$ by mapping $k$ to $x_k$.
	\item from a morphism $f$ from $B\sigma$ to $X$, construct the path $\pi = f(0) \xrightarrow{j_1,\alpha_1} \ldots \xrightarrow{j_n,\alpha_n} f(n)$.
\end{itemize}

%

\section{Omitted proofs of Section \ref{sec:unfo}}

\subsection*{Proof of Proposition \ref{prop:colimit}}

\begin{lemma}
$\col{D}$ is well-defined and is a pHDA.
\end{lemma}

\begin{proof}
\begin{enumerate}
	\item \textbf{$\partial_i^0(\eqcl{u,k})$ is well-defined.} Assume that $\eqcl{u,k} = \eqcl{v,k}$, $\alpha_k^u = \alpha_k^v = 0$ and that $i = j_k^u = j_k^v$. We want to prove that $\eqcl{u,k-1} = \eqcl{v,k-1}$. We actually prove something stronger: if $\alpha_k^u = 0$, then for every $\eqcl{u,k} = \eqcl{v,k}$, $\alpha_k^v = 0$, $j_k^v = j_k^u$ and $\eqcl{u,k-1} = \eqcl{v,k-1}$, by induction of the proof of $\eqcl{u,k-1} = \eqcl{v,k-1}$.
		\begin{itemize}
			\item the case where $k = 0$ is impossible, because $\alpha_0^u$ is not defined,
			\item if there is a morphism $\map{i}{u}{v}$ in $\C$, then this is a consequence of the fact that $Di$ is a morphism of pHDA,
			\item the last case is from the definition of $\approx$. Let us use the same notations. In this case, we necessarily have that $\alpha_k^{u'}$ is equal to $1$ for some $u'$ such that $(u,k) \sim (u',k)$. By the previous cases, $\alpha_k^{u'} = 0$. So this case is impossible.
		\end{itemize}
	\item \textbf{$\partial_i^1(\eqcl{u,k})$ is well-defined.} Assume that $\eqcl{u,k} = \eqcl{v,k}$, $\alpha_k^u = \alpha_k^v = 1$ and that $i = j_k^u = j_k^v$. We want to prove that $\eqcl{u,k+1} = \eqcl{v,k+1}$. This a particular case of the definition of $\approx$ with $l = 0$.
	\item \textbf{$\partial_{i_1 < \ldots < i_n}^{\alpha_1, \ldots, \alpha_n}$ is well-defined.} We have to prove that if $(j_1;\beta_1)\star\ldots\star(j_n;\beta_n) = (i_1 < \ldots < i_n ; \alpha_1, \ldots, \alpha_n) = (k_1;\gamma_1)\star\ldots\star(k_n;\gamma_n)$ and both $\partial_{j_n}^{\beta_n}\circ\ldots\circ\partial_{j_1}^{\beta_1}(\eqcl{u,k})$ and $\partial_{k_n}^{\gamma_n}\circ\ldots\circ\partial_{k_1}^{\gamma_1}(\eqcl{u,k})$ are well-defined, then they are equal. 
	
	First, observe that $\partial^0_j\circ\partial_i^1$ is never defined. Indeed, if $\partial^0_j\circ\partial_i^1(\eqcl{u,k})$ were defined, then that would mean there is $v$ with $\eqcl{v,k} = \eqcl{u,k}$ such that $\alpha_{k+1}^v = 1$ and $\partial_i^1(\eqcl{u,k})$ would be equal to $\eqcl{v,k+1}$. That would also mean that $\partial^0_j(\eqcl{v,k+1})$ is defined, that is, $\alpha_{k+1}^v = 0$, which is impossible.
	
	This means that, since $(j_1;\beta_1)\star\ldots\star(j_n;\beta_n) = (k_1;\gamma_1)\star\ldots\star(k_n;\gamma_n)$, there is $s$ such that for all $i \leq s$, $\gamma_i = \beta_i = 0$ and for all $i > s$, $\gamma_i = \beta_i = 1$. By what we proved in the previous point, necessarily, for every $i \leq s$, $j_s = k_s$. Consequently, we have $(j_s;1)\star\ldots\star(j_n;1) = (k_s;1)\star\ldots\star(k_n;1)$. It is enough to prove if $(j_1;1)\star\ldots\star(j_n;1) = (i_1 < \ldots < i_n ; 1, \ldots, 1) = (k_1;1)\star\ldots\star(k_n;1)$ and both $\partial_{j_n}^{1}\circ\ldots\circ\partial_{j_1}^{1}(\eqcl{u,k})$ and $\partial_{k_n}^{1}\circ\ldots\circ\partial_{k_1}^{1}(\eqcl{u,k})$ are well-defined, then they are equal. This is precisely what the definition of $\approx$ allows you to prove.

	\item \textbf{$\lambda$ is well defined.} If $\eqcl{u,k} = \eqcl{v,k}$, we need to prove that $w_k^u = w_k^v$, by induction on the proof that $\eqcl{u,k} = \eqcl{v,k}$:
		\begin{itemize}
			\item if $k = 0$, then $w_k^u = \epsilon = w_k^v$.
			\item if $\map{i}{u}{v} \in \C$, and if $k \leq l_u, l_v$, then $w_k^u = w_k^v$ because $Di$ is a morphism of pHDA.
			\item if we are in the case of the definition of $\approx$, then by induction hypothesis, $w_k^{u_0} = w_k^{v_0}$ and for every $0 \leq s \leq l$, $w_{k+s+1}^{u_s} = w_{k+s+1}^{u_{s+1}}$. Furthermore, by definition of a spine, $w_{k+s+1}^{u_s} = \delta_{j_k^{u_s}}^1(w_{k+s}^{u_s})$. Combining all those equalities, we obtain:
			$$w_{k+l+1}^{u_l} = \delta_{j_{k+l+1}^{u_l}}^{1}\circ\ldots\circ\delta_{j_{k+1}^{u_0}}^1(w_k^{u_0}) = \delta_{j_{k+l+1}^{v_l}}^{1}\circ\ldots\circ\delta_{j_{k+1}^{v_0}}^1(w_k^{u_0}) = \delta_{j_{k+l+1}^{v_l}}^{1}\circ\ldots\circ\delta_{j_{k+1}^{v_0}}^1(w_k^{v_0}) = w_{k+l+1}^{v_l}$$
		\end{itemize}
	\item \textbf{$\col{D}$ is a partial precubical set.} Given by the third condition of the definition of the face maps.
	\item \textbf{$\col{D}$ is a pHDA.} By definition of a spine.
\end{enumerate}
\end{proof}

Now, given an object $u$ of $\C$, we want to construct a map $\map{\iota_u}{Du}{\col{D}}$. This maps $k \in B\sigma_u$ to $\eqcl{u,k}$ (the dimension matches by definition of $B\sigma_u$).

\begin{lemma}
$\iota_u$ is a morphism of pHDA.
\end{lemma}

\begin{proof}~
\begin{enumerate}
	\item \textbf{$\iota_u$ is a morphism of partial precubical sets.}
\begin{itemize}
	\item if $\alpha_k^u = 0$, then $\partial_{j_k^u}^{\alpha_k^u}(k) = k-1$, then $$\iota_u(\partial_{j_k^u}^{\alpha_k^u}(k)) = \iota_u(k-1) = \eqcl{u,k-1} = \partial_{j_k^u}^{\alpha_k^u}(\eqcl{u,k}) = \partial_{j_k^u}^{\alpha_k^u}(\iota_u(k))$$
	\item similarly if $\alpha_k^u = 1$.
\end{itemize}
	\item \textbf{$\iota_u$ is a morphism of pHDA.} $\lambda\circ\iota_u(k) = \lambda(\eqcl{u,k}) = w_k^u = \lambda_{Du}(k)$.
\end{enumerate}
\end{proof}

\begin{proposition}
$(\col{D},\iota_u)$ is the colimit of D.
\end{proposition}

\begin{proof}~
\begin{enumerate}
	\item \textbf{$(\col{D},\iota_u)$ is a cocone.} Given a morphism $\map{i}{u}{v}$ in $\C$, we must prove that $\iota_v\circ Di = \iota_u$. But $Di(k) = k$, so $\iota_v\circ Di(k) = \eqcl{v,k}$. By definition of $\sim$, $[v,k] = [u,k]$ and so $\eqcl{v,k} = \eqcl{u,k}$.
	\item \textbf{$(\col{D},\iota_u)$ has the universal property.} Let $(Z,\kappa_u)$ be another cocone. 
		\begin{enumerate}
			\item \textbf{Construction of a morphism $\map{\Phi}{\col{D}}{Z}$.} Define $\Phi(\eqcl{u,k}) = \kappa_u(k)$ and $\Phi(\eqcl{\epsilon})$ being the initial state of $Z$.
			\item \textbf{$\Phi$ is well defined.} We want to prove that if $\eqcl{u,k} = \eqcl{v,k}$, then $\Phi(\eqcl{u,k}) = \Phi(\eqcl{v,k})$. We prove it by induction, and there are three cases to consider:
				\begin{itemize}
					\item $\Phi(\eqcl{u,0}) = \kappa_u(0)$ is the initial state of $Z$ since $\kappa_u$ is a morphism of pHDA. So $\Phi(\eqcl{u,0})$ coincide with $\Phi(\eqcl{\epsilon})$.
					\item If $\map{i}{u}{v} \in \C$, and if $k \leq l_u, l_v$, then since $(Z, \kappa_u)$ is a cocone, $\kappa_u(k) = \kappa_v\circ Di(k) = \kappa_v(k)$.
					\item If we are in the case of the definition of $\approx$, we have to prove that $\kappa_{u_l}(k+l+1) = \kappa_{v_l}(k+l+1)$. By induction hypothesis, $\kappa_{u_0}(k) = \kappa_{v_0}(k)$, and for every $0 \leq s < l$, $\kappa_{u_s}(k+s+1) = \kappa_{u_{s+1}}(k+s+1)$ and $\kappa_{v_s}(k+s+1) = \kappa_{v_{s+1}}(k+s+1)$. Furthermore, since $\kappa_x$ are morphisms of partial precubical sets, for every $0 \leq s \leq l$, 
					$$\kappa_{u_s}(k+s+1) = \kappa_{u_s}(\partial_{j_{k+s+1}^{u_s}}^1(k+s)) = \partial_{j_{k+s+1}^{u_s}}^1(\kappa_{u_s}(k+s))$$
					$$\kappa_{v_s}(k+s+1) = \kappa_{v_s}(\partial_{j_{k+s+1}^{v_s}}^1(k+s)) = \partial_{j_{k+s+1}^{v_s}}^1(\kappa_{v_s}(k+s))$$
					
					Combining all those equalities, we obtain:
					
					\begin{tabular}{rcl}
    $\kappa_{u_l}(k+s+1)$ & $=$ & $\partial_{j_{k+l+1}^{u_{l-1}}}^1\circ\ldots\circ\partial_{j_{k+1}^{u_0}}^1(\kappa_{u_0}(k))$\\
     & $=$ & $\partial_{j_{k+l+1}^{u_{l-1}}}^1\circ\ldots\circ\partial_{j_{k+1}^{u_0}}^1(\kappa_{v_0}(k))$\\
     & $=$ & $\partial_{j_{k+l+1}^{v_{l-1}}}^1\circ\ldots\circ\partial_{j_{k+1}^{v_0}}^1(\kappa_{v_0}(k))$\\
     & $=$ & $\kappa_{v_l}(k+s+1)$
\end{tabular}
				\end{itemize}
			\item \textbf{$\Phi$ is a morphism of pHDA.} Because the $\kappa_u$ are.
			\item \textbf{$\Phi$ is a morphism of cocones.} $\Phi(\iota_u(k)) = \Phi(\eqcl{u,k}) = \kappa_u(k)$.
			\item \textbf{$\Phi$ is unique.} Obvious.
		\end{enumerate}
\end{enumerate}
\end{proof}

\subsection*{Proof of Proposition \ref{prop:coltree}}

\begin{lemma}
\label{lem:pathcol}
Given $u$ and $k$, the following is a path in $\col{D}$:
$$\pi_{u,k} = \eqcl{u,0} \xrightarrow{j_1^u,\alpha_1^u} \ldots \xrightarrow{j_k^u,\alpha_k^u} \eqcl{u,k}$$
\end{lemma}

\begin{proof}
This is a rephrasing of the fact that $\iota_u$ is a morphism of pHDA.
\end{proof}

\begin{lemma}
\label{lem:eqch}
If $\eqcl{u,k} = \eqcl{v,k}$, then $\pi_{u,k} \sim_{ch} \pi_{v,k}$.
\end{lemma}

\begin{proof}
Let us prove it by induction on the proof of $\eqcl{u,k} = \eqcl{v,k}$. Three cases:
\begin{itemize}
	\item If $k = 0$, both paths are the empty path.
	\item If there is $\map{i}{u}{v}$, then $\pi_{u,k} = \pi_{v,k}$.
	\item If we are in the case of the definition of $\approx$, we have to prove that $\pi_{u_l,k+l+1} \sim_{ch} \pi_{v_l,k+l+1}$. By induction hypothesis, we have:
	$$\pi_{u_0,k} \sim_{ch} \pi_{v_0,k}$$
	and for every $0\leq s < l$, 
	$$\pi_{u_s,k+s+1} \sim_{ch} \pi_{u_{s+1},k+s+1}$$
	Let us denote by $\pi_u^s$ the following path in $\col{D}$:
	$$\pi_{u_0,k} \xrightarrow{j_{k+1}^{u_0}, 1} \eqcl{u_0,k+1} \xrightarrow{j_{k+2}^{u_1}, 1} \ldots \xrightarrow{j_{k+l+1}^{u_s}, 1} \eqcl{u_s,k+s+1}$$
	Idem for $\pi_v^s$. Let us prove by induction on $s$ that $\pi_u^s \sim_{ch} \pi_{u_s, k+s+1}$.
		\begin{itemize}
			\item[$\bullet$] \textit{Base case $s = 0$.} $\pi_u^0 = \pi_{u_0, k+1}$, by construction.
			\item[$\bullet$] \textit{Induction case.} 
			
\begin{tabular}{rclcr}
     $\pi_u^s$ & $=$ & $\pi_u^{s-1} \rightarrow{j_{k+s+1}^{u_s}} \eqcl{u_s,k+s+1}$ & ~ & by construction\\
     & $~~\sim_{ch}~~$ & $\pi_{u_{s-1},k+s} \rightarrow{j_{k+s+1}^{u_s}} \eqcl{u_s,k+s+1}$ & & induction hypothesis\\
     & $~~\sim_{ch}~~$ & $\pi_{u_{s},k+s} \rightarrow{j_{k+s+1}^{u_s}} \eqcl{u_s,k+s+1}$ & & see above\\
     & $=$ & $\pi_{u_s, k+s+1}$ & & by definition
\end{tabular}		
		\end{itemize}
So, we have $\pi_u^l \sim_{ch} \pi_{u_l, k+l+1}$, and similarly, $\pi_v^l \sim_{ch} \pi_{v_l, k+l+1}$. Define $\tilde{\pi}_u^l$ as the same path as $\pi_u^l$, except that $\pi_{u_0,k}$ is replaced by $\pi_{v_0,k}$, which are confluently homotopic, as seen above. Then $\pi_u^l \sim_{ch} \tilde{\pi}_u^l$. Finally, $ \tilde{\pi}_u^l \leftrightsquigarrow_{ch} \pi_v^l$, by construction.
\end{itemize}
\end{proof}

\begin{lemma}[Accessibility modulo confluent homotopy] 
\label{lem:acc}
Let $D$ be non-empty.
In $\col{D}$, for every path $\pi$, there is an object $u$ such that $\pi$ is confluently homotopic to a subpath of $\iota_u$, that is there is $u$ and $k$ such that $\pi \sim_{ch} \pi_{u,k}$.
\end{lemma}

\begin{proof}
Let us prove it by induction on the length of $\pi$. $\pi$ is of the form
$$\eqcl{u_0,0} \xrightarrow{i_1,\beta_1} \ldots \xrightarrow{i_k,\beta_k} \eqcl{u_k,k}$$
The $\eqcl{u_{k-1},k-1} \xrightarrow{i_k,\beta_k} \eqcl{u_k,k}$ part means that there is a $v$ such that:
\begin{itemize}
	\item $\eqcl{u_{k-1},k-1} = \eqcl{v,k-1}$ and $\eqcl{u_k,k} = \eqcl{v,k}$,
	\item $j_k^v = i_k$ and $\alpha_k^v = \beta_k$.
\end{itemize}
Let us denote by $\pi_{-1}$, the path $\pi$ from which we have removed the last transition. By induction hypothesis, $\pi_{-1} \sim_{ch} \pi_{u,k-1}$ for some $u$ and so $$\pi = \pi_{-1} \xrightarrow{i_k,\beta_k} \eqcl{v,k} \sim_{ch} \pi_{u,k-1} \xrightarrow{i_k,\beta_k} \eqcl{v,k}$$ In particular, this means that $\eqcl{v,k-1} = \eqcl{u,k-1}$. So, by Lemma \ref{lem:eqch}, $\pi_{u,k-1} \sim_{ch} \pi_{v,k-1}$. Consequently,
$$\pi = \pi_{-1} \xrightarrow{j_k^v,\alpha_k^v} \eqcl{v,k} \sim_{ch} \pi_{v,k-1} \xrightarrow{j_k^v,\alpha_k^v} \eqcl{v,k} = \pi_{v,k}$$
\end{proof}

\begin{proof}[of Proposition \ref{prop:coltree}]
\begin{enumerate}
	\item \textbf{Without shortcuts.} Given by the condition $$\partial_{j_1 < \ldots < j_m}^{\beta_1, \ldots, \beta_m}\circ\partial_{i_1 < \ldots < i_n}^{\alpha_1, \ldots, \alpha_n} \subseteq \partial_{k_1 < \ldots < k_{n+m}}^{\gamma_1, \ldots, \gamma_{n+m}}$$ of the definition of the face maps.
	\item \textbf{Existence of a path.} $\pi_{u,k}$ is a path to $\eqcl{u,k}$, by Lemma \ref{lem:pathcol}.
	\item \textbf{Unicity of the path.} Assume that there are two paths $\pi_1$ and $\pi_2$ to $\eqcl{u,k}$. By Lemma \ref{lem:acc}, there are $u_1$ and $u_2$ such that $\pi_i \sim_{ch} \pi_{u_i,k}$. This means in particular that $\eqcl{u_1,k} = \eqcl{u,k} = \eqcl{u_2,k}$, so by Lemma \ref{lem:eqch},
	$$\pi_1 \sim_{ch} \pi_{u_1,k} \sim_{ch} \pi_{u_2,k} \sim_{ch} \pi_2$$
\end{enumerate}
\end{proof}

\subsection*{Proof of Proposition \ref{prop:unftree}}

We need to construct a diagram $\map{D_X}{\C_X}{\ps}$, for which $U(X)$ is the colimit. Following ideas from \cite{dubut16}, it will essentially be given by the set of confluent homotopy classes of paths, but some extra care must be taken to handle confluent homotopy fully.

Define first $\C_X$ as the category whose:
\begin{itemize}
	\item objects are pairs $(u,L)$, where $L$ is a list $j_1, \ldots, j_n$ such that there is a path $\pi$ of $X$ that does not terminate with a transition of the form $\xrightarrow{~ \_, 1 ~}$ with 
	$$\pi \xrightarrow{~ j_1, 1 ~} \ldots \xrightarrow{~ j_n, 1 ~} x \in u$$
	\item morphisms are generated by the following. For every object $(u,L) \neq ([i], \epsilon)$, we have exactly one generator with codomain $(u,L)$ given by:
		\begin{itemize}
			\item either $L = L', j$, and we add a generator $\map{cat_{u,L',j}}{(v,L')}{(u,L)}$ for $v = [\pi_{-1}]$ with $\pi_{-1} \xrightarrow{~ j, 1 ~} x \in u$.
			\item either $L = \epsilon$. In this case, there is a $j$ such that every path of $u$ necessarily ends with a transition of the form $\xrightarrow{~j,0~}$. We \emph{choose} a unique object $(v,L')$ such that for every $\pi \in v$, $\pi \xrightarrow{~j,0~} e(u) \in u$, and add a generator $\map{cho_{u}}{(v,L')}{(u,\epsilon)}$.
		\end{itemize}
\end{itemize}

Next, define the diagram $D_X$ on objects $(u,L)$ by induction on the length of $u$ (we construct it with image in $\spine$ instead of $\ps$ for simplicity, identifying a spine with the path shape it induces), maintaining the fact that $e(D_X(u,L)) = (d,w)$, where $d$ is the dimension of $e(u)$, and $w$ is its labelling:
\begin{itemize}
	\item $D_X([i], \epsilon) = (0,\epsilon)$,
	\item if $u$ is of length $n+1$, and $L = L', j$, let $(v,L')$ as in the definition of $cat_{u,L',j}$. By induction hypothesis, $D_X(v,L')$ is constructed, and let $(d,w)$ be its ending. Define $$D_X(u,L) = D_X(v,L') \xrightarrow{~ j,1 ~} (d-1, \delta_j(w))$$
	\item if $u$ is of length $n+1$, and $L = \epsilon$, let $(v,L')$ and $j$ as in the definition of $cho_u$. By induction hypothesis, $D_X(v,L')$ is constructed, and let $(d,w)$ be its ending. Define $$D_X(u,L) = D_X(v,L') \xrightarrow{~ j,0 ~} (d+1, w')$$
with $w' = \lambda(e(u))$, which satisfies by construction that $\delta_j(w') = w$.
\end{itemize}
On morphisms, $D_X$ is defined as the obvious inclusions of spines.

We want to prove that $U(X)$ is the colimit of $D_X$. More precisely, we have to provide a morphism $\map{\mu_{u,L}}{D_X(u,L)}{U(X)}$ for every object $(u,L)$, such that $(U(X), (\mu_{u,L}))$ is the universal cocone of $D_X$. Again those morphisms are defined by induction on the length of $u$:
\begin{itemize}
	\item $\mu_{[i],\epsilon}(0,\epsilon) = [i]$,
	\item $\mu_{u,L',j}(y)$ is $\mu_{v,L'}(y)$ if $y \in D_X(v,L')$, or $u$, if $y = e(D_X(u,L))$,
	\item $\mu_{u,\epsilon}(y)$ is $\mu_{v,L'}(y)$ if $y \in D_X(v,L')$, or $u$, if $y = e(D_X(u,L))$.
\end{itemize}

\begin{lemma}
The $\mu_{u,L}$ are morphisms of pHDA.
\end{lemma}

\begin{proof}
This is done by induction on the length of $u$:
\begin{itemize}
	\item if $u = [i]$, obvious.
	\item if $L = L',j$, using the induction hypothesis, the only additional face defined in $D_X(u,L)$ which is not defined in $D_X(v, L')$ is $\partial_j^1(e(D_X(v,L')))$, which is equal to $e(D_X(u,L))$. Furthermore, in $U(X)$, $\partial_j^1(v) = u$ by construction.
	\item if $L = \epsilon$, using the induction hypothesis, the only additional face defined in $D_X(u,L)$ which is not defined in $D_X(v, L')$ is $\partial_j^0(e(D_X(u,L)))$, which is equal to $e(D_X(v,L))$. Furthermore, in $U(X)$, $\partial_j^0(u) = v$ by construction.
\end{itemize}
\end{proof}

\begin{proposition}
$(U(X), (\mu_{u,L}))$ is the colimit of $D_X$.
\end{proposition}

\begin{proof}
Let $(K, (\tau_{u,L}))$ be another cocone. Define $\map{\Phi}{U(X)}{K}$ as $$u \mapsto \tau_{u,L}(e(D_X(u,L)))$$
for any $L$ such that $(u,L)$ is an object of $\C_X$.
\begin{itemize}
	\item \textbf{This definition does not depends on the choice of $L$.} Two cases to consider:
		\begin{itemize}
			\item if $u = [i]$ or if any path of $u$ ends with a transition of the form $\xrightarrow{~\_,0 ~}$, then $L = \epsilon$ is the only choice possible.
			\item otherwise, assume $L = j_1, \ldots, j_n$ and $L' = k_1, \ldots, k_n$ are two such lists. Then necessarily, by definition of confluent homotopy, $$(j_1,1)\star \ldots \star (j_n,1) = (k_1,1)\star\ldots\star (k_n,1).$$
			and we can choose a confluent homotopy class of paths $v$ such that for every $\pi \in v$, $$\pi \xrightarrow{~j_1,1~} \ldots \xrightarrow{~j_n,1~} e(D_X(u,L)) \in u$$ and $$\pi \xrightarrow{~k_1,1~} \ldots \xrightarrow{~k_n,1~} e(D_X(u,L')) \in u$$
			The first point implies that $$\partial_{j_n}^1\circ\ldots\circ\partial_{j_1}^1 \equiv \partial_{k_n}^1\circ\ldots\circ\partial_{k_1}^1$$
			The second means that, by definition of $D_X(u,L)$, $$e(D_X(u,L)) = \partial_{j_n}^1\circ\ldots\circ\partial_{j_1}^1(D_X(v,\epsilon)) \text{~in~} D_X(u,L)$$ and $$e(D_X(u,L')) = \partial_{k_n}^1\circ\ldots\circ\partial_{k_1}^1(D_X(v,\epsilon))\text{~in~} D_X(u,L')$$
			So then, since $(K,(\tau_{u,L}))$ is a cocone made of morphisms of pHDA
\begin{tabular}{rclcr}
     $\tau_{u,L}(e(D_X(u,L)))$ & $=$ & $\partial_{j_n}^1\circ\ldots\circ\partial_{j_1}^1(\tau_{v,\epsilon}(D_X(v,\epsilon)))$ & ~ & morphism of pHDA\\
     & $=$ & $\partial_{k_n}^1\circ\ldots\circ\partial_{k_1}^1(\tau_{v,\epsilon}(D_X(v,\epsilon)))$ & & first point\\
     & $=$ & $\tau_{u,L'}(e(D_X(u,L')))$ & & morphism of pHDA
\end{tabular}	
		\end{itemize}
	
	\item \textbf{$\Phi\circ\mu_{u,L} = \tau_{u,L}$.} By induction on the length of $u$:
		\begin{itemize}
			\item if $u = [i]$, $L = \epsilon$, $$\Phi\circ\mu_{[i],\epsilon}(0,\epsilon) = \Phi([i]) = ]tau_{[i],\epsilon}(e(D_X([i],\epsilon))) = \tau_{[i],\epsilon}(0,\epsilon)$$
			\item if $u$ is of length $n+1$, and $L = L',j$, two cases:
				\begin{itemize}
					\item $y \in D_X(v,L')$, use the induction hypothesis.
					\item $y = e(D_X(u,L))$, $$\Phi\circ\mu_{u,L}(y) = \Phi(u) = \tau_{u,L}(e(D_X(u,L)))$$
				\end{itemize} 
			\item if $u$ is of length $n+1$, and $L = \epsilon$, the proof is the same. For the first case, use the $L'$ from the definition of $cho_u$.
		\end{itemize}
	
	\item \textbf{$\Phi$ is a morphism of pHDA.} Two cases to consider:
		\begin{itemize}
			\item if we have a path $\pi$ of the form $\pi' \xrightarrow{~j,1~} x$, that is $\partial_j^1([\pi']) = [\pi]$, we have to prove that $\partial_j^1(\Phi([\pi'])) = \Phi([\pi])$. Then we can choose $L'$ such that $([\pi],L',j)$ and $([\pi'],L')$ are objects of $\C_X$. Then:\\

\begin{center}
\begin{tabular}{rclcr}
     $\partial_j^1(\Phi([\pi']))$& $=$ & $\partial_j^1(\tau_{[\pi'],L'}(e(D_X([\pi'],L'))))$ & & definition of $\Phi$\\
     & $=$ & $\partial_j^1(\tau_{[\pi],L}(e(D_X([\pi'],L'))))$ & & $\tau$ is a cocone\\
     & $=$ & $\tau_{[\pi],L}(\partial_j^1(e(D_X([\pi'],L'))))$ & & $\tau_{[\pi],L}$ morphism of pHDA\\
     & $=$ & $\tau_{[\pi],L}(e(D_X([\pi],L)))$ & & definition of $D_X([\pi],L)$\\
     & $=$ & $\Phi([\pi])$ & & definition of $\Phi$
\end{tabular}
\end{center}
			\item if we have a path $\pi$ of the form $\pi' \xrightarrow{~j,0~} x$, that is $\partial_j^0([\pi]) = [\pi']$, we have to prove that $\partial_j^0(\Phi([\pi])) = \Phi([\pi'])$. Then we can choose $L'$ from the definition of $cho_{[\pi]}$, and we have that $([\pi'], L')$ and $([\pi],\epsilon)$ are objects of $\C_X$. Then:\\

\begin{center}
\begin{tabular}{rclcr}
     $\partial_j^0(\Phi([\pi]))$& $=$ & $\partial_j^0(\tau_{[\pi],\epsilon}(e(D_X([\pi],L))))$ & & definition of $\Phi$\\
     & $=$ & $\tau_{[\pi],\epsilon}(\partial_j^0(e(D_X([\pi],\epsilon))))$ & & $\tau_{[\pi],\epsilon}$ morphism of pHDA\\
     & $=$ & $\tau_{[\pi'],L'}(\partial_j^0(e(D_X([\pi],\epsilon))))$ & & $\tau$ cocone\\
     & $=$ & $\tau_{[\pi'],L'}(e(D_X([\pi'],L')))$ & & definition of $D_X([\pi],\epsilon)$\\
     & $=$ & $\Phi([\pi'])$ & & definition of $\Phi$
\end{tabular}
\end{center}
		\end{itemize}
		
	\item \textbf{$\Phi$ is unique.} Assume there is another $\Phi'$, then $\Phi'(u) = \Phi'\circ\mu_{u,L}(e(D_X(u,L)))$ for any $L$ such that $(u,L)$ is an object of $\C_X$. Then $\Phi'(u) = \tau_{u,L}(e(D_X(u,L))) = \Phi(u)$.
\end{itemize}
\end{proof}

\subsection*{Proof of Proposition \ref{prop:unf}}

~

\begin{enumerate}
	\item \textbf{$\unf_X$ is a morphism of partial precubical sets.} Two cases to consider:
		\begin{itemize}
			\item If $\partial_i^1([\pi])$ is defined then $$\unf_X(\partial_i^1([\pi])) = \unf_X([\pi \xrightarrow{i,1} \partial_i^1(e(\pi))]) =  \partial_i^1(e(\pi)) = \partial_i^1(\unf_X([\pi]))$$
			\item If $\partial_i^0(\alpha)$ is defined then $\alpha = [\pi = \pi_{-1} \xrightarrow{i,0} e(\alpha)]$, and so $$\unf_X(\partial_i^0(\alpha)) = \unf_X([\pi_{-1}]) = e(\pi_{-1}) = \partial_i^0(e(\alpha)) = \partial_i^0(\unf_X(\alpha))$$
		\end{itemize}
	\item \textbf{$\unf_X$ is a morphism of pHDA.} By construction of the labelling function of the unfolding.
	\item \textbf{$\unf_X$ is a covering.} Assume given a path 
	$$\Gamma = H_0 \xrightarrow{j_1,\alpha_1} H_1 \xrightarrow{j_2,\alpha_2} \ldots \xrightarrow{j_n,\alpha_n} H_n$$
	in $U(X)$. Now consider a path $\pi'$ of $X$, of the form
	$$\unf_X(\Gamma) \xrightarrow{i,\alpha} x$$
	The goal is to prove that there is a unique $H$, such that
	$$\Gamma' = H_0 \xrightarrow{j_1,\alpha_1} H_1 \xrightarrow{j_2,\alpha_2} \ldots \xrightarrow{j_n,\alpha_n} H_n \xrightarrow{i,\alpha} H$$
	is a path of $U(X)$, and
	$$\unf_X(\Gamma') = \pi'$$
	Two cases two consider:
	\begin{itemize}
		\item If $\alpha = 0$, $\partial_i^0(x) = e(H_n)$. Define $H = [\pi \xrightarrow{i,0} x]$ for $\pi \in H_n$. Then $\Gamma' = \Gamma \xrightarrow{i,0} H$ is a path in $U(X)$ such that $\unf_X(\Gamma') = \pi'$. Actually, this is the unique such $\Gamma'$ since $H$ is the unique possible such that $\Gamma \xrightarrow{i,0} H$ is a path.
		\item If $\alpha = 1$, $x = \partial_i^1(e(H_n))$. Defining $H = [\pi \xrightarrow{i,1} x]$ works too.
	\end{itemize}
	\item \textbf{$\unf_X$ is universal.} By Proposition \ref{prop:coltree}.
\end{enumerate}

\subsection*{Proof of Proposition \ref{prop:natiso}}

\begin{enumerate}
	\item \textbf{Definition of $\map{\eta_X}{X}{U(X)}$.} $\eta_X(x) = \xi$, where $\xi$ is the unique class of paths modulo confluent homotopy from the initial state to $x$.
	\item \textbf{$\eta_X$ is a morphism of partial precubical sets.} Let $x$ in $X$ such that $\partial_{i_1 < \ldots < i_k}^{\alpha_1, \ldots, \alpha_k}(x)$ is defined. Since $X$ is without shortcuts, there is a sequence $$(j_1;\beta_1)\star\ldots\star(j_n;\beta_n) = (i_1 < \ldots < i_n;\alpha_1, \ldots, \alpha_n)$$ such that $\partial_{j_n}^{\alpha_n}\circ\ldots\circ\partial_{j_1}^{\alpha_1}(x)$ is defined, and so is equal to $\partial_{i_1 < \ldots < i_k}^{\alpha_1, \ldots, \alpha_k}(x)$. Let us prove by induction that $\partial_{j_i}^{\alpha_i}\circ\ldots\circ\partial_{j_1}^{\alpha_1}(\eta_X(x))$ is defined and is equal to $\eta_X(\partial_{j_i}^{\alpha_i}\circ\ldots\circ\partial_{j_1}^{\alpha_1}(x))$.
		\begin{itemize}
			\item[$\bullet$] \textit{Base case $i = 0$.} $\eta_X(x) = \eta_X(x)$.
			\item[$\bullet$] \textit{Induction case.} Assume that $\partial_{j_i}^{\alpha_i}\circ\ldots\circ\partial_{j_1}^{\alpha_1}(\eta_X(x))$ is defined and is equal to $\eta_X(\partial_{j_i}^{\alpha_i}\circ\ldots\circ\partial_{j_1}^{\alpha_1}(x))$. In particular, $$e(\partial_{j_i}^{\alpha_i}\circ\ldots\circ\partial_{j_1}^{\alpha_1}(\eta_X(x))) = \partial_{j_i}^{\alpha_i}\circ\ldots\circ\partial_{j_1}^{\alpha_1}(x).$$ Two cases:
				\begin{itemize}
					\item $\alpha_{i+1} = 1$: Then $\partial_{j_{i+1}}^{\alpha_{i+1}}\circ\partial_{j_i}^{\alpha_i}\circ\ldots\circ\partial_{j_1}^{\alpha_1}(\eta_X(x))$ is defined as $$[\pi \xrightarrow{j_{i+1},1} \partial_{j_{i+1}}^{\alpha_{i+1}}\circ\partial_{j_i}^{\alpha_i}\circ\ldots\circ\partial_{j_1}^{\alpha_1}(x)]$$
					for any $\pi \in \eta_X(\partial_{j_i}^{\alpha_i}\circ\ldots\circ\partial_{j_1}^{\alpha_1}(x))$.
					\item $\alpha_{i+1} = 0$: Since $\partial_{j_{i+1}}^0(\partial_{j_i}^{\alpha_i}\circ\ldots\circ\partial_{j_1}^{\alpha_1}(x))$ is defined, there is a path $\pi$ from the initial state to it. So, $$\pi \xrightarrow{j_{i+1},0} \partial_{j_i}^{\alpha_i}\circ\ldots\circ\partial_{j_1}^{\alpha_1}(x)$$
					is a path from the initial state to $\partial_{j_i}^{\alpha_i}\circ\ldots\circ\partial_{j_1}^{\alpha_1}(x)$. By unicity of path modulo confluent homotopy, this path belongs to $\partial_{j_i}^{\alpha_i}\circ\ldots\circ\partial_{j_1}^{\alpha_1}(\eta_X(x))$, and so $\partial_{j_{i+1}}^{\alpha_{i+1}}\circ\partial_{j_i}^{\alpha_i}\circ\ldots\circ\partial_{j_1}^{\alpha_1}(\eta_X(x))$ is defined as $[\pi]$.
				\end{itemize}
		\end{itemize}
Consequently, $\eta_X(\partial_{i_1 < \ldots < i_k}^{\alpha_1, \ldots, \alpha_k}(x))$ is defined and is equal to $\partial_{i_1 < \ldots < i_k}^{\alpha_1, \ldots, \alpha_k}(\eta_X(x))$
	\item \textbf{$\eta_X$ is a morphism of pHDA.} $\lambda(\eta_X(x)) = \lambda(\xi) = \lambda(e(\xi))$. But by definition of $\xi$, $e(\xi) = x$.
	\item \textbf{$\eta_X$ is the inverse of $\unf_X$.}
		\begin{itemize}
			\item $\eta_X\circ\unf_X([\pi]) = \eta_X(e(\pi))$. But $[\pi]$ and $\eta_X(e(\pi))$ are both confluent homotopy classes of paths from the initial state to $e(\pi)$. By unicity, we get that they are equal and $\eta_X\circ\unf_X([\pi]) = [\pi]$.
			\item $\unf_X\circ\eta_X(x) = e(\xi) = x$ by definition of $\xi$.
		\end{itemize}
	\item \textbf{$\eta$ is natural.} Let $\map{f}{X}{Y}$. $U(f)\circ\eta_X$ maps $x$ to the image by $f$ of the unique homotopy class of paths from the initial state to $x$. $\eta_Y\circ f$ maps $x$ to the unique homotopy class of paths from the initial state to $f(x)$. By unicity, they both coincide.
\end{enumerate}

\subsection*{Proof of Corollary \ref{theo:coref}}

We prove that $\unf$ is the counit and $\eta$ is the unit.
\begin{enumerate}
	\item \textbf{$\unf$ is natural.} $f\circ\unf_X$ and $\unf_Y\circ U(f)$ both maps an homotopy class of paths to their common end point.
	\item \textbf{First unit-counit equation.} $\unf_X\circ\eta_X = \text{id}_X$ by the previous Proposition.
	\item \textbf{Second unit-counit equation.} To prove that $U(\unf_X)\circ\eta_{U(X)} = \text{id}_{U(X)}$, it is enough to prove that $U(\unf_X) = \unf_{U(X)}$. To this end, let us observe that, since the unfolding is a tree, then by Lemma \ref{lem:acc}, any path in $U(X)$ from the initial state to $[\pi]$, with $$\pi = x_0 \xrightarrow{j_1,\alpha_1} x_1 \xrightarrow{j_2,\alpha_2} \ldots \xrightarrow{j_n,\alpha_n} x_n$$ is of the form $$\Gamma = [\pi_{-n}] \xrightarrow{j_1,\alpha_1} \ldots \xrightarrow{j_{n-1},\alpha_{n-1}} [\pi_{-1}] \xrightarrow{j_n,\alpha_n} [\pi]$$ where $$\pi_{-k} = x_0 \xrightarrow{j_1,\alpha_1} x_1 \xrightarrow{j_2,\alpha_2} \ldots \xrightarrow{j_{n-k},\alpha_{n-k}} x_{n-k}$$ So $$U(\unf_X)([\Gamma]) = [x_0 \xrightarrow{j_1,\alpha_1} x_1 \xrightarrow{j_2,\alpha_2} \ldots \xrightarrow{j_n,\alpha_n} x_n] = [\pi] = \unf_{U(X)}([\Gamma])$$
\end{enumerate}

\section{Omitted proofs of Section \ref{sec:cofi}}

\subsection*{Proof of Lemma \ref{lem:retract}}

\begin{lemma}
\label{lem:mocoho}
If $\map{f}{X}{Y}$ is a morphism of pHDA, and if $\pi \sim_{ch} \pi'$, then $f(\pi) \sim_{ch} f(\pi')$.
\end{lemma}

\begin{proof}
Easy induction on the proof that $\pi \sim_{ch} \pi'$.
\end{proof}

\begin{proof}[of Lemma \ref{lem:retract}]
Consider a section $\map{s}{X}{Y}$ and a retract $\map{r}{Y}{X}$, that is, $r\circ s = \text{id}_X$.
\begin{enumerate}
	\item \textbf{If $Y$ is without shortcuts then $X$ is without shortcuts.} Let $x \in X$ such that $\partial_{i_1 < \ldots < i_n}^{\alpha_1, \ldots, \alpha_n}(x)$ is defined. Then, since $s$ is a morphism, $\partial_{i_1 < \ldots < i_n}^{\alpha_1, \ldots, \alpha_n}(s(x))$ is defined and is equal to $s(\partial_{i_1 < \ldots < i_n}{\alpha_1, \ldots, \alpha_n}(x))$. Since $Y$ is without shortcuts, then there is a sequence $$(j_1;\beta_1)\star\ldots\star(j_n;\beta_n) = (i_1 < \ldots < i_n;\alpha_1, \ldots, \alpha_n)$$ such that $\partial_{j_n}^{\alpha_n}\circ\ldots\circ\partial_{j_1}^{\alpha_1}(s(x))$ is defined. Since $r$ is a morphism, $$\partial_{j_n}^{\alpha_n}\circ\ldots\circ\partial_{j_1}^{\alpha_1}(x) = \partial_{j_n}^{\alpha_n}\circ\ldots\circ\partial_{j_1}^{\alpha_1}(r(s(x)))$$ is defined.
	\item \textbf{If every element of $Y$ is accessible, then every element of $X$ is accessible.} Let $x$ be such an element. Then since $Y$ is a tree and $s(x)$ is accessible, that is, there is a path $\pi$ to $s(x)$. Then $r(\pi)$ is path to $x$.
	\item \textbf{If $Y$ is has the unique path property modulo confluent homotopy, then $X$ has the unique path property modulo confluent homotopy.} Let $\pi$ and $\pi'$ be two paths to $x$. Then $s(\pi) \sim_{ch} s(\pi')$ since $Y$ is a tree. Then $\pi = r(s(\pi)) \sim_{ch} r(s(\pi')) = \pi'$ by Lemma \ref{lem:mocoho}.
\end{enumerate}
\end{proof}

\subsection*{Proof of Proposition \ref{prop:cofi}}

Let $X$ be a tree and such a diagram
		\begin{center}
	\begin{tikzpicture}[scale=1.5]
		\node (A) at (0,0.9) {$\ast$};
		\node (B) at (1.2,0.9) {$Y$};
		\node (C) at (0,0) {$X$};
		\node (D) at (1.2,0) {$Z$};
		\path[->,font=\scriptsize]
		(A) edge node[above]{$!$} (B)
		(B) edge node[right]{$f$} (D)
		(C) edge node[below]{$g$} (D)
		(A) edge node[left]{$!$} (C);
	\end{tikzpicture}
	\end{center}
	Let $X_n$ be the restriction of $X$ to elements of depth smaller than $n$ and such that the face maps are such that $\partial_{k_1 < \ldots < k_{n}}^{\gamma_1, \ldots, \gamma_{n}}(x)$ is define in $X_n$ iff there is a sequence $(i_1;\alpha_1)\star\ldots\star(i_n;\alpha_n) = (k_1 < \ldots < k_n ; \gamma_1, \ldots, \gamma_n)$ such that for every $s$, $\partial_{i_s}^{\alpha_s}\circ\ldots\circ\partial_{i_1}^{\alpha_1}(x)$ is defined in $X$ and belongs to $X_n$.
		We will construct by induction on $n$ map $\map{h_n}{X_n}{Y}$ such that 
			\begin{center}
	\begin{tikzpicture}[scale=1.5]
		\node (A) at (0,0.9) {$\ast$};
		\node (B) at (1.2,0.9) {$Y$};
		\node (C) at (0,0) {$X_n$};
		\node (D) at (1.2,0) {$Z$};
		\path[->,font=\scriptsize]
		(A) edge node[above]{$!$} (B)
		(B) edge node[right]{$f$} (D)
		(C) edge node[below]{$g$} (D)
		(A) edge node[left]{$!$} (C);
		\path[->,font=\scriptsize,dotted]
		(C) edge node[above]{$h_n$} (B);
	\end{tikzpicture}
	\end{center}
	and for all $i < n$, $h_{n}$ and $h_i$ coincide on $X_i$. This is enough to construct the lift that is needed, because the inductive limit of the $X_n$ is $X$, since $X$ is without shortcuts.
	
\begin{itemize}
	\item \textit{Base case.} $\map{h_0}{X_0}{Y}$ is the unique map that sends the initial state of $X$ to the initial state of $Y$.
	\item \textit{Induction case.} Assume that $h_0$, \ldots, $h_n$ are constructed.
	
	\begin{enumerate}
		\item \textbf{Construction of $h_{n+1}$.} Let $x \in X_{n+1}\setminus X_n$. Let $\alpha$ be the unique confluent homotopy class of paths to $x$. Two cases:
	\begin{itemize}
		\item either there is a unique $i$ such that every path $\pi$ of $\alpha$ ends with $\partial_i^0(x) \xrightarrow{i,0} x$. In particular, $\partial_i^0(x) \in X_n$. Define $h_{n+1}(x)$ as follows.
		There is a commutative diagram of the form:
		\begin{center}
	\begin{tikzpicture}[scale=1.5]
		\node (A) at (0,0.9) {$B\sigma_\pi$};
		\node (B) at (1.2,0.9) {$Y$};
		\node (C) at (0,0) {$B\sigma_\pi'$};
		\node (D) at (1.2,0) {$Z$};
		\path[->,font=\scriptsize]
		(A) edge node[above]{$h_n\circ\pi$} (B)
		(B) edge node[right]{f} (D)
		(C) edge node[below]{$g\circ\pi'$} (D)
		(A) edge node[left]{$!$} (C);
		\path[->,font=\scriptsize, dotted]
		(C) edge node[above]{$\theta$} (B);
	\end{tikzpicture}
	\end{center}
	where $\pi$ is a path to $\partial_i^0(x)$ and $\pi' = \pi \xrightarrow{i,0} x$. Since $f$ is open, there is a map $\map{\theta}{B\sigma_\pi'}{Y}$ that makes the two triangles commute. Then define $h_{n+1}(x) = \theta(n+1)$. We can check that $$g(x) = f\circ h_{n+1}(x) = f(\theta(n+1))$$ by the lower triangle.
		\item or every path $\pi$ of $\alpha$ ends with $x' \xrightarrow{i,1} \partial_i^1(x') = x$ for some $x'$ and $i$, not necessarily unique. In particular, all such $x'$ belong to $X_n$. Define $h_{n+1}(x)$ as $\partial_i^1(h_n(x'))$ for any such $x'$. We must check that $\partial_i^1(h_n(x'))$ is defined. There is a commutative diagram:
		\begin{center}
	\begin{tikzpicture}[scale=1.5]
		\node (A) at (0,0.9) {$B\sigma_{\pi_{-1}}$};
		\node (B) at (1.2,0.9) {$Y$};
		\node (C) at (0,0) {$B\sigma_\pi$};
		\node (D) at (1.2,0) {$Z$};
		\path[->,font=\scriptsize]
		(A) edge node[above]{$h_n\circ\pi_{-1}$} (B)
		(B) edge node[right]{f} (D)
		(C) edge node[below]{$g\circ\pi$} (D)
		(A) edge node[left]{$!$} (C);
		\path[->,font=\scriptsize, dotted]
		(C) edge node[above]{$\theta$} (B);
	\end{tikzpicture}
	\end{center}
	and since $f$ is open, there is a map $\map{\theta}{B\sigma_\pi}{Y}$ that makes the two triangles commute. Then $\theta(n+1)$ is necessarily equal to $\partial_i^1(h_n(x'))$, and the latter is defined. We must then check that the value of $h_{n+1}$ does not depend on the choice of $x'$. Let $\pi_1$ and $\pi_2$ be two paths to $x$, $\pi_s$ ending with $x_s' \xrightarrow{i_s,1} x$. We want to prove that $\partial_i^1(h_n(x_1')) = \partial_i^1(h_n(x_2'))$ (we already know they exist). If $\pi_1 \leftrightsquigarrow_{ch} \pi_2$ (the general case $\pi_1 \sim_{ch} \pi_2$ follows easily), there is $y$ in $X$ and two sequences $(j_{1,1},1)\star\ldots\star(j_{k,1},1)\star(i_1,1) = (j_{1,2},1)\star\ldots\star(j_{k,2},1)\star(i_2,1)$ such that $\pi_s$ ends with $y \xrightarrow{j_{1,s},1} \ldots \xrightarrow{j_{k,s},1} x_s' \xrightarrow{i_s,1} x$. Since $h_n$ is a map of pHDA, $\partial_{j_{k,s}}^1\circ\ldots\circ\partial_{j_{1,s}}^1(h_n(y)) = h_n(x_s')$. So $$\partial_{i_1}^1(x_1') = \partial_{i_1}^1\circ\partial_{j_{k,1}}^1\circ\ldots\circ\partial_{j_{1,1}}^1(h_n(y)) = \partial_{i_2}^1\circ\partial_{j_{k,2}}^1\circ\ldots\circ\partial_{j_{1,2}}^1(h_n(y)) = \partial_{i_2}^1(x_2')$$
	and $h_{n+1}(x)$ is uniquely defined. Next, we must check that $$f(h_{n+1}(x)) = f(\partial_i^1(h_n(x'))) = \partial_i^1(f(h_n(x'))) = \partial_i^1(g(x')) = g(\partial_i^1(x')) = g(x).$$
	\end{itemize}
	
	\item \textbf{$h_{n+1}$ is a morphism of pHDA.} We need to prove that if $x \in X_{n+1}$ and if $\partial_{i_1 < \ldots < i_n}^{\alpha_1, \ldots, \alpha_n}(x)$ is defined in $X_{n+1}$ then $\partial_{i_1 < \ldots < i_n}^{\alpha_1, \ldots, \alpha_n}(h_{n+1}) = h_{n+1}(\partial_{i_1 < \ldots < i_n}^{\alpha_1, \ldots, \alpha_n}(x))$. By definition of a face map being defined in $X_{n+1}$, it is enough to prove that if $x$ is of depth $\leq n+1$, and if $\partial_i^\alpha(x)$ is defined in $X$ and of depth $\leq n+1$, then $h_{n+1}(\partial_i^\alpha(x)) = \partial_i^\alpha(h_{n+1}(x))$. Three cases to consider:
\begin{itemize}
	\item $x$ and $\partial_i^\alpha(x)$ are of depth $\leq n$: then it comes from the fact that $h_n$ is a map of pHDA.
	\item if $x$ is of depth $n$, $\partial_i^\alpha(x)$ of depth $n+1$ and so $\alpha = 0$: in this case, $h_{n+1}(x) = h_n(x)$ and $h_{n+1}(\partial_i^\alpha(x))$ is defined using the first construction. The equation is a consequence of the upper commutative triangle.
	\item if $x$ is of depth $n+1$, $\partial_i^\alpha(x)$ of depth $n$ and so $\alpha = 1$: the equation is exactly the definition of $h_{n+1}(x)$ in the second construction.
\end{itemize}
		
	\end{enumerate}

\end{itemize}

If furthermore we assume that $f$ is a $\ps$-covering, and if we call $h'$ another lift, then let us prove by induction on $n$ that, for all $n$, $h'_n$, the restriction of $h'$ on $X_n$ coincides with $h_n$. 
\begin{itemize}
	\item \textit{Base case.} $h_0$ and $h'_0$ are both the morphisms that maps the initial state of $X$ to the initial state of $Y$.
	\item \textit{Induction case.} It depends on the case used in the construction of $h_{n+1}$:
		\begin{itemize}
			\item In the first case, with the same notations, since $\partial_i^0(x)$ is in $X_n$, $h'_n(\partial_i^0(x)) = h_n(\partial_i^0(x))$ by induction hypothesis. Since $h'_{n+1}$ is a morphism of pHDA, then $h'_n(\partial_i^0(x)) = \partial_i^0(h'_{n+1}(x))$. So $\theta' = \theta_{-1} \xrightarrow{i,0} h'_{n+1}(x)$ is another lift. Since $f$ is a covering, then $\theta = \theta'$, and in particular, $h_{n+1}(x) = h'_{n+1}(x)$.
			\item In the second case, 

\begin{center}			
\begin{tabular}{rclcr}
    $h'_{n+1}(x)$ & $=$ & $h'_{n+1}(\partial_i^1(x'))$ & ~ & by assumption\\
     & $=$ & $\partial_i^1(h'_n(x'))$ & & $h'$ is a morphism of pHDA\\
     & $=$ & $\partial_i^1(h_n(x'))$ & & by induction hypothesis\\
     & $=$ & $h_{n+1}(x)$ & & by construction of $h_{n+1}$
\end{tabular}
\end{center}
		\end{itemize}
\end{itemize}

\subsection*{Proof of Corollary \ref{prop:univ}}

\begin{enumerate}
	\item \textbf{Existence and unicity of $h$.} Since $Y$ is tree and since $g$ is a covering, then by Proposition \ref{prop:cofi}, there is a unique such $h$.
	\item \textbf{$h$ is open.} Let a diagram of the form
\begin{center}
	\begin{tikzpicture}[scale=1.5]
		\node (A) at (0,0.9) {$B\sigma$};
		\node (B) at (1.2,0.9) {$Y$};
		\node (C) at (0,0) {$B\sigma'$};
		\node (D) at (1.2,0) {$Z$};
		\path[->,font=\scriptsize]
		(A) edge node[above]{$x$} (B)
		(B) edge node[right]{$h$} (D)
		(C) edge node[below]{$y$} (D)
		(A) edge node[left]{$p$} (C);
	\end{tikzpicture}
	\end{center}
	
Then we have the following diagram (in plain)
	\begin{center}
	\begin{tikzpicture}[scale=1.5]
		\node (A) at (0,0.9) {$B\sigma$};
		\node (B) at (1.2,0.9) {$Y$};
		\node (C) at (0,0) {$B\sigma'$};
		\node (D) at (1.2,0) {$X$};
		\path[->,font=\scriptsize]
		(A) edge node[above]{$x$} (B)
		(B) edge node[right]{$f = g\circ h$} (D)
		(C) edge node[below]{$g\circ y$} (D)
		(A) edge node[left]{$p$} (C);
		\path[->,font=\scriptsize,dotted]
		(C) edge node[above]{$\theta$} (B);
	\end{tikzpicture}
	\end{center}
	Since $f$ is open, there is $\map{\theta}{B\sigma'}{U(X)}$ with $\theta\circ p = x$ and $g\circ y = g\circ h\circ\theta$. Since $g$ is a covering, $y = h\circ\theta$ and $\theta$ is the lifting we were looking for.
	\item \textbf{$h$ is a covering.} Assume that we have two lifts.
	\begin{center}
		\begin{tikzpicture}[scale=1.5]
		\node (A) at (0,0.9) {$B\sigma$};
		\node (B) at (1.2,0.9) {$Y$};
		\node (C) at (0,0) {$B\sigma'$};
		\node (D) at (1.2,0) {$Z$};
		\path[->,font=\scriptsize]
		(A) edge node[above]{$x$} (B)
		(B) edge node[right]{$h$} (D)
		(C) edge node[below]{$y$} (D)
		(A) edge node[left]{$p$} (C);
		\path[->,font=\scriptsize,dotted]
		(C) edge node[above]{$\theta_1, \theta_2$} (B);
	\end{tikzpicture}
	\end{center}
	Since $h\circ\theta_1 = h\circ\theta_2$, then $f\circ\theta_1 = f\circ\theta_2$. Since $f$ is a covering, then $\theta_1= \theta_2$.
\end{enumerate}

\end{document}